\documentclass{article}

\usepackage{amsmath,amsthm,amssymb,cite,enumerate} 

\usepackage{color,graphicx}

\def \d {\mathrm{d}}

\def \T {\mbox{\ensuremath{\bigtriangleup}}}
\def \Rabcd {{\stackrel{\scriptscriptstyle{(0)}}{R}}{}^{ ab}_{\phantom{ab}cd}}
\def \Rceag {{\stackrel{\scriptscriptstyle{(0)}}{R}}{}^{ ce}_{\phantom{ce}ag}}
\def \Rbedg {{\stackrel{\scriptscriptstyle{(0)}}{R}}{}^{ be}_{\phantom{be}dg}}
\def \C {{\stackrel{\scriptscriptstyle{(-2)}}{C}}} 

\newcommand{\sgn}{\operatorname{sgn}}
\newcommand{\s}{\operatorname{s}}
\newcommand{\cs}{\operatorname{c}}
\newcommand{\ct}{\operatorname{ct}}

\newtheorem{thm}{Theorem}[section] 
\newtheorem{cor}[thm]{Corollary} 
\newtheorem{lem}[thm]{Lemma} 
\newtheorem{prop}[thm]{Proposition}
\newtheorem{con}[thm]{Conjecture} 
\theoremstyle{definition} 
\newtheorem{defn}[thm]{Definition}
  
\theoremstyle{remark}  
  
\def\beq{\begin{eqnarray}}  
\def\eeq{\end{eqnarray}}  
\def\bsp{\begin{split}}  
\def\esp{\end{split}}

\def\d{\mathrm{d}}  

\def\T{ {\sf T} }

\def \eq5d {\stackrel{{\mbox{\tiny (n=5)}}}{=}}

\newcommand{\mb}[1]{{\mathbb #1}}

\def \bl {\mbox{{\mbold\ell}}}
\def \bn {\mbox{{\bf n}}}

\def \bT {\mbox{{\bf T}}}
\def \bg {\mbox{{\bf g}}}

\def \BEA { \begin{eqnarray}}
\def \EEA {\end{eqnarray}}
\def \BE {\begin{equation}}
\def \EE {\end{equation}}
\def\d{\mathrm{d}}

\def \mi {\stackrel{i}{m}}
\def \mj {\stackrel{j}{m}}
\def \mk {\stackrel{k}{m}}
\def \mr {\stackrel{r}{m}}
\def \ms {\stackrel{s}{m}}
\def \mz {\stackrel{z}{m}}
\def \mq {\stackrel{q}{m}}
\def \mo {\stackrel{o}{m}}
\def \mD {\stackrel{2}{m}}
\def \mT {\stackrel{3}{m}}
\def \mC {\stackrel{4}{m}}

\def \mio #1 {\mi_{#1}\ ^{  \! \! \! \! 0}} 
\def \mjo #1 {\mj_{#1}\ ^{  \! \! \! \! 0}} 
\def \mko #1 {\mk_{#1}\ ^{  \! \! \! \! 0}} 
\def \mro #1 {\mr_{#1}\ ^{  \! \! \! \! 0}} 
\def \mso #1 {\ms_{#1}\ ^{  \! \! \! \! 0}} 
\def \mpo #1 {\mp_{#1}\ ^{  \! \! \! \! 0}} 
\def \mzo #1 {\mz_{#1}\ ^{  \! \! \! \! 0}} 
\def \mqo #1 {\mq_{#1}\ ^{  \! \! \! \! 0}} 
\def \moo #1 {\mo_{#1}\ ^{  \! \! \! \! 0}} 
\def \mDo #1 {\mD_{#1}\ ^{  \! \! \! \! 0}} 
\def \mTo #1 {\mT_{#1}\ ^{  \! \! \! \! 0}} 
\def \mCo #1 {\mC_{#1}\ ^{  \! \! \! \! 0}} 

\def \miJ #1 {\mi_{#1}\ ^{  \! \! \! \! (1)}} 
\def \mjJ #1 {\mj_{#1}\ ^{  \! \! \! \! (1)}} 
\def \mkJ #1 {\mk_{#1}\ ^{  \! \! \! \! (1)}} 
\def \mrJ #1 {\mr_{#1}\ ^{  \! \! \! \! (1)}}

\def \bl {\mbox{\boldmath{$\ell$}}}
\def \hbl {\mbox{\boldmath{$\hat \ell$}}}
\def \bn {\mbox{\boldmath{$n$}}}

\def \hbn {\mbox{\boldmath{$\hat n$}}}
\def \hbm #1 {\mbox{\boldmath{$\hat m^{(#1)}$}}}

\def \Mi {\stackrel{i}{M}}

\def \Ms {\stackrel{s}{M}}

\def \Mchal {\stackrel{\chi_\alpha}{M}}

\newcommand{\be}{\begin{equation}}
\newcommand{\ee}{\end{equation}}
\newcommand{\beqn}{\begin{eqnarray}}
\newcommand{\eeqn}{\end{eqnarray}}

\newcommand{\ba}{\begin{array}}
\newcommand{\ea}{\end{array}}

\def \Mi {\stackrel{i}{M}}

\def \Ms {\stackrel{s}{M}}
\def \pul {\frac{1}{2}}

\def \T {\bigtriangleup  }

\newcommand{\M}[3] {{\stackrel{#1}{M}}_{{#2}{#3}}}

\newcommand{\Phia}{\Phi^\mathrm{A}}
\newcommand{\Phis}{\Phi^\mathrm{S}} 
\def \mb #1 {\mbox{\boldmath{$m^{(#1)}$}}}

\begin{document}

\title{Type II universal spacetimes}

\author{S. Hervik$^\diamond$, T. M\'alek$^\star$,  V. Pravda$^\star$,  A. Pravdov\' a$^\star$\\
\vspace{0.05cm} \\
{\small $^\diamond$ Faculty of Science and Technology, University of Stavanger}, {\small  N-4036 Stavanger, Norway}  \\
{\small $^\star$ Institute of Mathematics, Academy of Sciences of the Czech Republic}, \\ {\small \v Zitn\' a 25, 115 67 Prague 1, Czech Republic} \\
 {\small E-mail: \texttt{sigbjorn.hervik@uis.no, malek@math.cas.cz,
pravda@math.cas.cz, pravdova@math.cas.cz}} }

\date{\today}

\maketitle

\begin{abstract}
We study type II  universal metrics of the Lorentzian signature. These metrics simultaneously solve vacuum field equations of all theories of gravitation with the Lagrangian being a polynomial curvature invariant constructed from the metric, the Riemann tensor and its covariant derivatives of an arbitrary order. 

We provide examples of type II  universal metrics for all composite number dimensions. 
On the other hand, we have no examples for  prime number dimensions and we prove the non-existence of type II  universal spacetimes in five dimensions.

We also present type II vacuum solutions of selected classes of gravitational theories, such as  Lovelock, quadratic and $L(\mbox{Riemann})$ gravities.

\end{abstract}
\section{Introduction}

Let us start with a formal definition of universal  \cite{Coleyetal08} and $k$-universal metrics.
\begin{defn}
\label{univ}
A metric is called {\it $k$-universal} if all conserved symmetric rank-2 tensors constructed\footnote{{  Throughout the paper, we consider only scalars and symmetric rank-2 tensors  constructed as contractions of {\em polynomials} 
from  the  metric, the Riemann tensor and its covariant  derivatives
of an arbitrary order. }}  
from  the  metric, the Riemann tensor and its covariant  derivatives
up to the $k^{\textnormal{th}}$ order are multiples of the metric.
If a  metric is $k$-universal for all integers $k$ then it is called {\it universal}.
\end{defn}

Universal metrics are  vacuum solutions (possibly with a non-vanishing cosmological constant) to all theories of the form 
\be
{ L}={ L}(g_{ab},R_{abcd},\nabla_{a_1}R_{bcde},\dots,\nabla_{a_1\dots a_p}R_{bcde}),\label{Lagr}
\ee
where the Lagrangian $L$ is a polynomial curvature invariant.

Particular explicit examples of universal metrics were first discussed in \cite{AmaKli89,HorSte90}  in the context of string theory  (pp-waves) and in \cite{Coleyetal08} as metrics with vanishing quantum corrections. In \cite{universal1}, we have systematically studied type N and III  universal metrics (while all type N
Einstein spacetimes are 0-universal, they are universal iff in addition they are Kundt). In the present paper, we focus on the type II (including the type D) in the classification of \cite{Coleyetal04} (see also \cite{OrtPraPra12rev} for a recent review of the classification).

While the well known examples of type N universal pp-waves of \cite{AmaKli89,HorSte90} admit a covariantly constant null vector (CCNV)  obeying $\ell_{a;b}=0$, we show that in arbitrary dimension type II   universal spacetimes admitting CCNV do not exist. On the other hand, we construct examples of type II universal spacetimes admitting a recurrent null vector (RNV) obeying $\ell_{a;b} \propto \ell_a \ell_b$ and also more general Kundt non-RNV type II  universal spacetimes.

Interestingly, in contrast with the results of \cite{universal1} for  type III and N universal metrics, the type II cases critically depend on dimensionality of the spacetime.
We show that one can employ $n_0$-dimensional type N universal metrics of \cite{universal1} ($n_0 \geq 4$) to construct  type II universal metrics in $n= n_0 N$ dimensions (where $N\geq 2$ is an integer) using an appropriate direct product with $(N-1)$ maximally symmetric spaces. Simple direct product metrics can be also used to construct type D universal metrics for all composite number dimensions $n \geq 4$. 

On the other hand, these methods cannot be used for prime number dimensions and  we prove that 
\begin{thm}
\label{prop5d}
In five dimensions,  genuine type II  universal spacetimes do not exist.
\end{thm}
In fact, in section \ref{sec_5d}, we prove a stronger statement that in five dimensions,  genuine type II  $0$-universal spacetimes do not exist.

These results might indicate the non-existence of type II  universal metrics for prime number dimensions. However, 
one should keep in mind that the five-dimensional case is special (see eq.\ \eqref{Weyl5D}).

Since the direct product metrics discussed above provide examples of universal metrics  for  dimensions $n=N n_0$, where $n_0\geq 4$, $N\geq 2$, 
for composite number of dimensions, it
remains to find examples for $n=N n_0$, where $n_0=2,3$.
We thus proceed with studying
 type II generalized Khlebnikov--Ghanam--Thompson metrics in $n=n_0 N$ dimensions, where { $n_0 \geq 2$}. In the $n_0=2,3$ cases, we find necessary and sufficient conditions under 
which these spacetimes are $0$-universal
and thus  solve all vacuum gravitational field equations 
that  do not contain derivatives of the Riemann tensor (Lovelock gravity). In fact, since $0$-universal spacetimes are Einstein
they also solve theories that
may contain derivatives of the Ricci tensor, such as e.g.\ quadratic gravity.

For $n_0=2$, we also find a necessary and sufficient condition for $2$-universality. These spacetimes solve
all theories containing up to the second derivatives of the Riemann tensor 
(e.g.\ $L($Riemann$)$ theories). We  also discuss further necessary and sufficient conditions for 4-universality and universality.

The paper is organized as follows:

In section \ref{sec_prelim}, we briefly summarize the notation and some  useful definitions and results. 

In section \ref{sec_conserved}, we prove that two appropriate rank-2 tensors   are conserved for universal spacetimes. These conserved tensors are then used in section \ref{sec_5d} to prove the non-existence of type II  0-universal spacetimes in five dimensions. 
 
In section \ref{sec_recurrent},  
we prove  that Ricci-flat recurrent type II 0-universal spacetimes
do not exist in any number of dimension.  

In section \ref{sec_examples}, we prove universality for various type II  direct product metrics. 

In section \ref{sec_GTmain}, we study generalizations of the Khlebnikov--Ghanam--Thompson metric (representing gravitational radiation in the (anti-)Nariai spacetime) and necessary and sufficient conditions following from various gravitational theories and from universality.

Finally, some useful relations for the Weyl tensor of direct product metrics and spin coefficients and curvature of generalized Khlebnikov--Ghanam--Thompson metrics are given in the appendices.


\section{Preliminaries}
\label{sec_prelim}

Throughout the paper, we use the algebraic classification of tensors \cite{Coleyetal04}, recently also reviewed in \cite{OrtPraPra12rev}.
By the type II of the Weyl tensor, we will mean the {\em genuine} type II, which includes the type D but does not include the types III and N (see \cite{OrtPraPra12rev} for the definitions). 

We employ the higher-dimensional Newman--Penrose and Geroch--Held--Penrose formalisms \cite{Durkeeetal10}. 
For GHP quantities, we  follow the notation of \cite{Durkeeetal10}, for other NP quantities we use the notation summarized in \cite{OrtPraPra12rev}.

A null frame in $n$ dimensions  consists of null vectors $\bl$ and $\bn$ and $n-2$ spacelike vectors $\mbox{\boldmath{$m^{(i)}$}} $ obeying
\be
\ell^a \ell_a= n^a n_a = 0, \qquad   \ell^a n_a = 1, \qquad \ m^{(i)a}m^{(j)}_a=\delta_{ij}.  \label{ortbasis}
\ee
The coordinate indices $a,b, \ldots $ and frame indices $i,j,\ldots$ take values from 0 to $n-1$ and 2 to $n-1$, respectively.

The Lorentz transformations between null frames are generated by boosts
\be
\hbl = \lambda \bl, \qquad    \hbn = \lambda^{-1} \bn, \qquad   \hbm{i} = \mbox{\boldmath{$m^{(i)}$}},  \label{boost}
\ee
null rotations and spins.
We say that a quantity $q$ has a {boost weight} (b.w.) ${\rm b}$ if it transforms under a boost~\eqref{boost} according to
\be
\hat q = \lambda^{\rm b} q . 
\ee 
Various components of a tensor in the  null frame  may have distinct integer b.ws. 
{{The boost order}} of a tensor with respect to a null frame  is the maximum b.w. of its (non-vanishing) frame components. 
The boost order of a tensor depends only on the null direction $\bl$  (see e.g.\ Proposition 2.1 in \cite{OrtPraPra12rev}).   

In general, the boost order of the Weyl (and the Riemann) tensor is 2. However, in this paper, we   focus on type II (and D) spacetimes 
that by definition, admit at least one multiple Weyl aligned null direction (mWAND)\footnote{Spacetimes admitting more than one mWAND are of the type D.}.  
If we identify our frame
vector $\bl$ with an mWAND b.w.\ +2 and +1 components of the Weyl tensor vanish. Furthermore, it has been shown in \cite{DurRea09} that 
in a type II Einstein spacetime, at least one of the mWANDs is geodetic.

The type II Weyl tensor admits a frame decomposition \cite{Coleyetal04}
\BEA
   C_{abcd}\! &=&\!
	\phantom{+} 
	\overbrace{ 4 \Phi\, \, n^{}_{\{a} \ell^{}_{ b} n^{}_{ c} \ell^{}_{\, d\: \}} 
\;  + \;  8 \Phia_{ij} 
\, \, n^{}_{\{a} \ell^{}_{ b} m^{(i)}_{\, c} m^{(j)}_{\, d\: \}} 
      +8 \Phi_{ij} 
			\, \, n^{}_{\{a} m^{(i)}_{\, b} \ell^{}_{c} m^{(j)}_{\, d\: \}}
   +  \Phi_{ijkl}\, \, m^{(i)}_{\{a} m^{(j)}_{\, b} m^{(k)}_{\, c} m^{(l)}_{\, d\: \}}}^{\mathrm{\scriptscriptstyle{ boost\ weight\ 0  }}}
	\nonumber\\
      &&+  \overbrace{
    8 \Psi'_{i} 
		\, \ell^{}_{\{a} n^{}_b \ell^{}_c m^{(i)}_{\, d\: \}} +
    4 \Psi'_{ijk} 
		\, \ell^{}_{\{a} m^{(i)}_{\, b} m^{(j)}_{\, c} m^{(k)}_{\, d\: \}}}^{\scriptscriptstyle{ {-1   
     }}} 
  +\overbrace{
      4 \Omega'_{ij} 
			\, \ell^{}_{\{a} m^{(i)}_{\, b}  \ell^{}_{c}  m^{(j)}_{\, d\: \}}}^{\scriptscriptstyle{ { -2       
      }  }} ,\label{eq:rscalars} 
\EEA
where for an arbitrary tensor $T_{abcd}$ we define
\be
T_{\{ a bc d\} } \equiv \pul (T_{[a b] [c d]}+ T_{[c d] [ab]})  
\label{Weyl_symm}
\ee
and thus $C_{abcd}= C_{\{abcd\}}$.

Independent b.w.\ 0 components of the Weyl tensor are  $\Phia_{ij}$ and $\Phi_{ijkl}$ since
\be
\Phis_{ij} = -\pul\Phi_{ikjk},\ \ \Phi=\Phi_{ii},
\ee
where $\Phis_{ij}$ and $\Phia_{ij}$ are the symmetric and antisymmetric parts of $\Phi_{ij}$, respectively, and  
 $\Phi_{ijkl}$ has the `Riemann tensor' symmetries
\be
\Phi_{ijkl} = \Phi_{[ij][kl]} = \Phi_{klij}, \quad \quad  \Phi_{i[jkl]}=0.
\ee
 For the type II, b.w.\ negative components
of the Weyl tensor, $\Psi'_{ijk}$ and $\Omega'_{ij}$, are  also present in general.

When $n>4$, $\Phi_{ijkl}$ can be decomposed like the  Riemann tensor in an auxiliary $(n-2)$-dimensional Riemannian space as \cite{ColHer10}
\be
	\Phi_{ijkl}=\tilde\Phi_{ijkl}-\frac{4}{n-4}\left( \delta_{i[k}\Phis_{l]j}-\delta_{j[k}\Phis_{l]i}\right)
+\frac{4}{(n-3)(n-4)}\Phi\delta_{i[k}{\delta}_{l]j} , \qquad \tilde\Phi_{ijkj}=0 \qquad (n>4), \label{decompPhi4}
\ee
where $\Phi$ is the trace of $\Phi_{ij}$. Note that in five dimensions, $\tilde\Phi_{ijkl}$ vanishes identically.

Let us identify the frame vector $\bl$ with a geodetic mWAND and without loss of generality,  choose $\bl$ to be affinely parameterized and the remaining frame vectors to be parallelly transported along $\bl$.
Then, the covariant derivatives of the frame vectors in terms of spin coefficients read  
\BEA
\ell_{a ; b} &=& L_{11} \ell_a \ell_b  + L_{1i} \ell_a m^{(i)}_{\, b}  +
\tau_i 
m^{(i)}_a \ell_b    + {\rho}_{ij} m^{(i)}_{\, a} m^{(j)}_{\, b}  , \label{dl} \\
n_{a ; b  } &=&\! -\! L_{11} n_a \ell_b -\! L_{1i} n_a m^{(i)}_{\, b}  +
\kappa'_i 
 m^{(i)}_{\, a} \ell_b   + \rho'_{ij} 
m^{(i)}_{\, a} m^{(j)}_{\, b}  , \label{dn} \\
m^{(i)}_{a ; b } &=&\! -\! \kappa'_i 
\ell_a \ell_b  -\tau_i 
n_a \ell_b  
-\rho'_{ij} 
\ell_a m^{(j)}_{\, b}   
 \! +\! {\Mi}_{j1} m^{(j)}_{\, a} \ell_b  -\rho_{ij} n_a m^{(j)}_{\, b} 
+ {\Mi}_{kl} m^{(k)}_{\, a} m^{(l)}_{\, b} . \label{dm} 
\EEA

For Einstein spacetimes,   $R_{ab}= (R/n) g_{ab}$, with $R=$const., 
\BE
C_{abcd} = R_{abcd} - \frac{2 R}{ n (n-1)}  g_{a[c} g_{d]b} \label{RiemannWeyl}
\EE
{and thus }
\BE
C_{abcd;e} = R_{abcd;e}.
\EE
Therefore, for universal spacetimes { (that are necessarily Einstein)}, all conserved symmetric rank-2 tensors
constructed from the Weyl tensor and its derivatives are also proportional to the metric.

Let us conclude this section with the theorem proven in \cite{universal1}:
\begin{thm}
\label{prop_univCSI}
{ A universal  spacetime is necessarily a CSI spacetime.} 
\end{thm}
Constant/vanishing scalar curvature invariants (CSI/VSI) spacetimes are spacetimes 
having all curvature invariants constructed from the Riemann tensor and its derivatives constant \cite{ColHerPel06}/vanishing \cite{Coleyetal04vsi}.
Note that being CSI (and even VSI) is not a sufficient condition for universality.

\section{Conserved tensors}
\label{sec_conserved}

In the proof of   theorem \ref{prop5d} in section  \ref{sec_5d}, two rank-2  tensors, $S^{(2)}_{ab}$ and $S^{(3)}_{ab}$, are 
employed. In this section, we show that these two tensors are conserved for universal spacetimes. In fact, we prove the conservation
of these two tensors  for a  more general class of CCCT spacetimes.
\begin{defn}
A spacetime is CCCT (a covariantly constant conserved tensor)  if all conserved symmetric rank-2 tensors $T_{ab}$ constructed  
from  the  metric, the Riemann tensor and its covariant  derivatives
of arbitrary order obey $T_{ab;c}=0$.   
\end{defn}

Obviously, for universal spacetimes, U, U $\subset$ CCCT. 
In \cite{universal1} (see the end of section 3 there), we have proven a slightly more general version of theorem \ref{prop_univCSI}
stating that CCCT $\subset$ CSI.  

From the conservation of the Einstein tensor, it follows that for CCCT spacetimes, 
the  Ricci scalar is constant (this also obviously follows from the CCCT $\subset$ CSI result) and thus for curvature tensors,
 $R_{ab;c}=0$ and  
\BE
R_{abcd;e}=C_{abcd;e}.\label{CCCT_RIEM}
\EE

\begin{lem}
\label{lemmaCCconserved}
For CCCT spacetimes, $S^{(2)}_{ab} \equiv C_{acde} C_{b}^{\ cde}$ is a conserved tensor.
\end{lem}
\begin{proof}



Using the Bianchi identities and the tracelessness of the Weyl tensor,  
\BEA
S^{(2)  b}_{\ \ a \   ;b}&=& C_{acde;b} C^{bcde} + C_{acde} C^{bcde}_{\phantom{bcde} ;b}= C_{acde;b} C^{bcde}=
-C_{deba;c} C^{bcde}-C_{decb;a} C^{bcde}\nonumber\\
&=&
-C_{abde;c} C^{cbde}+C_{bcde;a} C^{bcde}
\ \ \rightarrow\ \ C_{acde;b} C^{bcde}=\frac{1 }{2 }  C_{bcde;a} C^{bcde}. \label{consS2}
\EEA
Since CCCT $\subset$ CSI, it follows that the last term in \eqref{consS2} vanishes due to the constancy of $C_{abcd} C^{abcd}$.
\end{proof}

\begin{lem}
\label{lemmaCCCconserved}
For CCCT spacetimes, $S^{(3)a}_{\ \ \ \ b} \equiv C^{ac}_{\ \ d e} C^{de}_{\ \ fg}  C^{fg}_{\ \ bc}$ is a conserved tensor. 
\end{lem}
\begin{proof}

$S^{(3)a}_{\ \ \ \ b}$ can be rewritten as  
\[ C^{ac}_{~~A}C^A_{~B}C^B_{~bc} \ ,\]
where $a,b,c,... $ are spacetime indices, while $A,B,C,...$ are bivector indices (essentially, $A=de$ etc.) for simplicity. 

First, let us  express the covariant derivative of the invariant 
\be
0= (C^A_{~B}C^B_{~C}C^C_{~A})_{;a}=3C^A_{~B}C^B_{~C}C^C_{~A;a} \ , \label{difinvar}
\ee
where the first equality follows from CCCT $\subset$ CSI.

We proceed with expressing $S^{(3)a}_{\ \ \ \ b;a}$ 
\be
S^{(3)a}_{\ \ \ \ b;a}=(C^{ac}_{~~A}C^A_{~B}C^B_{~bc})_{;a}=C^{ac}_{~~A;a}C^A_{~B}C^B_{~bc}+C^{ac}_{~~A}C^A_{~B;a}C^B_{~bc}+C^{ac}_{~~A}C^A_{~B}C^B_{~bc;a} \ .
\label{S3divergence}
\ee
For CCCT spacetimes, the first term on the right-hand side  vanishes due to the Bianchi identities and the tracelessness of the Weyl tensor.

 The 3rd term on the right-hand side is equal to (after renaming indices and using the Weyl symmetries)
\be 
C^{ac}_{~~A}C^A_{~B}C^B_{~bc;a}=-C^{ac}_{~~A}C^A_{~B}C^B_{~ab;c}-C^{ac}_{~~A}C^A_{~B}C^B_{~ca;b}=-C^{ac}_{~~A}C^A_{~B}C^B_{~bc;a}+C^{ac}_{~~A}C^A_{~B}C^B_{~ac;b} \ ,
\ee
which implies (using also \eqref{difinvar})
\be
 C^{ac}_{~~A}C^A_{~B}C^B_{~bc;a}=\frac 12 C^{C}_{~~A}C^A_{~B}C^B_{~C;b}=0.
\ee
Next, in the 2nd term in eq. \eqref{S3divergence} ($B=ef$)
\beq 
C^{ac}_{~~A}C^A_{~ef;a}C^{ef}_{~~bc}=-C^{ac}_{~~A}C^A_{~ae;f}C^{ef}_{~~bc}-C^{ac}_{~~A}C^A_{~fa;e}C^{ef}_{~~bc} =-2 C^{ac}_{~~A}C^A_{~ae;f}C^{ef}_{~~bc} \ , \label{eqCCCaux0}
\eeq
the term  
$C^{ac}_{~~A}C^A_{~ae;f}$
can be written as
\be
 C^{ac}_{~~A}C^A_{~ae;f}=\left(C^{ac}_{~~A}C^A_{~ae}\right)_{;f}-C^{ac}_{~~A;f}C^A_{~ae} \ . \label{eqCCCaux1}
\ee
By lemma \ref{lemmaCCconserved}, for CCCT spacetimes, the first term on the right-hand side in \eqref{eqCCCaux1} is zero, hence this expression is antisymmetric in $ce$ (after lowering $c$)
\be
C^{a}_{~cA}C^A_{~ae;f}=C^{a}_{~~[c|A}C^A_{~a|e];f} \ .\label{antisCC}
\ee
In eq. \eqref{eqCCCaux0}, \eqref{antisCC} is multiplied by $C^{ef}_{~~bc}$ (which is antisymmetric in $ef$) and so by defining 
\[ T_{[cef]}= C^{a}_{~~[c|A}C^A_{~a|e;f]} \ , \] 
we have 
\[ C^{a}_{~~cA}C^A_{~ae;f}C^{efc}_{~~~~b}=T_{[cef]}C^{efc}_{~~~~b} \ .\]
Therefore, since $C^{[efc]}_{~~~~b}=0$ this is zero and the right-hand side of eq.\ \eqref{eqCCCaux0} vanishes and thus the right-hand side of \eqref{S3divergence} vanishes as well. 
\end{proof}

\section{Type II universal spacetimes in five dimensions do not exist}
\label{sec_5d}

In this section, we prove theorem \ref{prop5d}, i.e.\ we show that type II universal spacetimes in five dimensions do not exist.

In five dimensions, all b.w.\ zero components of the Weyl tensor are determined by $\Phi_{ij}$ since  $\Phi_{ijkl}$
 can be expressed in terms of $\Phis_{ij}$ as \cite{PraPraOrt07} (see also \eqref{decompPhi4})
\be
	\Phi_{ijkl}\eq5d 4\left( \delta_{i[l}\Phis_{k]j}-\delta_{j[l}\Phis_{k]i}\right)
+2\Phi\delta_{i[k}{\delta}_{l]j}.\label{Weyl5D} 
\ee

As has been already pointed out in
section \ref{sec_prelim},  for universal spacetimes, all conserved symmetric rank-2 tensors
constructed from the Weyl tensor and its derivatives are also proportional to the metric.
By lemma \ref{lemmaCCconserved}, the rank-2 tensor $C_{acde} C_b^{\  cde} $ is conserved for universal spacetimes and thus for universal spacetimes
\be
S^{(2)}_{ab}=C_{acde} C_b^{\  cde} = K g_{ab}, \label{QGterm}
\ee
with $K$ being a constant.
The contraction of this equation with the frame vectors $\ell^a n^b$ and $m^{a}_{(i)} m^{b}_{(j)}$ 
(employing \eqref{Weyl5D}) gives
\be
2 \Phi \Phis_{ij} -3 \Phia_{ik} \Phia_{jk} - \Phis_{ik} \Phis_{jk} + \delta_{ij} \left(2 \Phis_{kl} \Phis_{kl} -\Phi^2    \right)=\frac{K}{4} \delta_{ij}, \label{eqtypeIIQ}
\ee
where
\be
K=2 \left(\Phi^2 - 3\Phia_{ij} \Phia_{ij} +  \Phis_{ij} \Phis_{ij}\right). 
\ee
Let us choose a frame with diagonal $\Phis_{ij}$ (not necessarily parallelly propagated). Clearly,  eq.\ \eqref{eqtypeIIQ} implies that the off-diagonal components of  $\Phia_{ik} \Phia_{jk}$ vanish and thus  two of the three independent components of $\Phia_{ij}$ vanish. Without loss of generality, we take $\Phis_{ij}$ and $\Phia_{ij}$ in the form
\be
\Phis_{ij} = \left( \begin{array}{ccc}
p_1 & 0 & 0 \\
0 & p_2 & 0 \\
0 & 0 & p_3 \end{array} \right),  \quad
\Phia_{ij} = \left( \begin{array}{rcc}
0 & a & 0 \\
-a & 0 & 0 \\
0 & 0 & 0 \end{array} \right).
\ee
Then, eq.\ \eqref{eqtypeIIQ}  reads
\BEA
 {p_1}^2 - 3 p_2 p_3 - p_1 p_2 - p_1 p_3  = 0, \label{eqtypeIIp1} \\
{p_2}^2 - 3 p_1 p_3 -p_1 p_2    - p_2 p_3 = 0, \label{eqtypeIIp2} \\
 {p_3}^2   - 3 p_1 p_2 - p_1 p_3 - p_2 p_3 + 3 a^2 = 0. \label{eqtypeIIp3}
\EEA
Note that by subtracting the first two equations we get
\be
(p_1-p_2) (p_1+p_2+2 p_3) = 0.
\ee
The only non-trivial solution of \eqref{eqtypeIIp1}--\eqref{eqtypeIIp3}  is
\be
\Phis_{ij} = \left( \begin{array}{ccc}
p & 0 & 0 \\
0 & p & 0 \\
0 & 0 & 0 \end{array} \right),  \quad
\Phia_{ij} = \left( \begin{array}{rrc}
0 & \pm p & 0 \\
\mp p & 0 & 0 \\
0 & 0 & 0 \end{array} \right), \label{5dsolQ}
\ee
which satisfies
\be
\Phis_{ij}  \Phis_{jk} = \frac{\Phi}{2} \Phis_{ik}, \quad
\Phia_{ij}  \Phis_{jk} = \frac{\Phi}{2} \Phia_{ik}, \quad
\Phia_{ij}  \Phia_{jk} = -\frac{\Phi}{2} \Phis_{ik}. \label{5DPhiproducts}
\ee

Now, let us use the necessary condition for universal spacetimes  from lemma \ref{lemmaCCCconserved} 
\be
S^{(3)}_{ab}=C^{cdef}C_{cdga}  C_{ef\ b}^{\ \  g} = K' g_{ab}, \ \ \ K'=\mbox{const}. \label{un3C}
\ee
Contracting  this equation with  $\mb{i} $, $\mb{j} $, using  \eqref{Weyl5D} and \eqref{5DPhiproducts},  leads  (after long but straightforward calculations) to 
\be
S^{(3)}_{ij}=12\Phi^2 \Phis_{ij} =K' \delta_{ij},
\ee
which holds only for $p=0=K'$. 
Together with \eqref{5dsolQ}, this implies that $\Phi_{ij}$ (and thus also all b.w.\ zero components of the Weyl tensor) vanishes which concludes the proof of theorem \ref{prop5d}. Since only conditions \eqref{QGterm} and  \eqref{un3C} were used, type II 0-universal spacetimes  do not exist in five dimensions.

\section{Recurrent type II universal spacetimes}
\label{sec_recurrent}

In this section, we study recurrent spacetimes,  i.e.\ spacetimes admitting 
a RNV 
\be
\ell_{a;b} = L_{11} \ell_a \ell_b.\label{recur}
\ee 
By comparing \eqref{recur} with \eqref{dl}, it follows that $\rho_{ij}$ vanishes and recurrent spacetimes 
thus belong to the Kundt class (see e.g.\ \cite{OrtPraPra12rev}). In fact, such spacetimes coincide
with the $\tau_i=0$ subclass of the Kundt metrics (see eq. (45) in \cite{universal1}). As pointed out in \cite{OrtPraPra07}, for Einstein Kundt spacetimes,
$\bl$ is an mWAND and thus the Weyl tensor is of the type II or more special .

For recurrent Einstein spacetimes, the  Newman--Penrose equation  (A5) of \cite{Durkeeetal10} reduces to
\be
 \Phi_{ij} = - \frac{R }{n(n-1)}  \delta_{ij} =  \frac{\Phi }{n-2} \delta_{ij} , \label{KundtRNV}
\ee
where   
\be
\Phi = -\frac{d-2}{d-1} \Lambda. \label{PhiLambda}
\ee
Projections of the necessary condition for universal spacetimes \eqref{QGterm}  
onto  $(\bl$, $\bn)$ and  $(\mb{i} $, $\mb{j} )$ planes give
\be
K = \frac{2 (n-1)}{n-2} \Phi^2
\ee
and
\be
\Phi_{iklm} \Phi_{jklm} = \frac{2n (n-3)}{(n-2)^2} \Phi^2 \delta_{ij}, \label{Phiiklmcond} 
\ee
respectively.

In the $\Lambda=0$ case, equation  \eqref{PhiLambda} implies $\Phi=0$ and from \eqref{KundtRNV} and the trace of \eqref{Phiiklmcond},
it then follows  that all b.w. zero components of the Weyl tensor vanish and thus in arbitrary dimension 
\begin{prop}
\label{prop_recurrent}
There are no recurrent Ricci-flat genuine type II 0-universal spacetimes.
\end{prop}
However, recurrent Ricci-flat type III and N  universal spacetimes do exist \cite{universal1}.
Note that pp-waves, i.e.\ spacetimes admitting a 
CCNV, are obviously recurrent and since Einstein  CCNV spacetimes
are  Ricci flat,  
{\it type II  0-universal pp-waves do not exist}.
Note also that proposition \ref{prop_recurrent} cannot be generalized to proper Einstein spacetimes -
in section \ref{sec_examples}, we provide examples of recurrent type II universal Einstein spacetimes with $\Lambda\not=0$.

Using \eqref{decompPhi4}, we can express eq.\ \eqref{Phiiklmcond} in terms of $\tilde\Phi_{jklm}$ as
\be
\tilde \Phi_{iklm} \tilde \Phi_{jklm} = \frac{2 (n-1)^2 (n-4)}{(n-2)^2 (n-3)} \Phi^2 \delta_{ij}. 
\ee
 The left-hand side of the above equation identically vanishes in four and five dimensions. In four dimensions, vanishing of  the right-hand side is guaranteed by the factor $(n-4)$, while in five dimensions this implies 
\be
\Phi = 0 \quad {\mbox{for}} \quad  n=5.
\ee
Consequently, for Einstein type II recurrent spacetimes obeying \eqref{QGterm},  all b.w.\ 0 components of the Weyl tensor vanish. 
For Einstein spacetimes, the field equations for quadratic gravity imply \eqref{QGterm} \cite{MalekPravdaQG}\footnote{ Note that for quadratic gravity,  $K$ is not  constant in general, however, it is constant  for Einstein recurrent spacetimes.} and thus 
\begin{prop}
In five dimensions, there are no genuine type II or D recurrent Einstein vacuum solutions of quadratic gravity.
\end{prop}
Therefore, in the recurrent case, the condition \eqref{QGterm} is sufficient to exclude the existence of five-dimensional 0-universal solutions.

\section{Examples of type II universal spacetimes}
\label{sec_examples}

In this section, we construct explicit examples of type II universal spacetimes.

\subsection{Universal type D product manifolds }

Let us consider  $N$ manifolds ($M_0$, $g^{(0)}_{a_0 b_0}$),  ($M_1$, $g^{(1)}_{a_1 b_1}$), \dots ($M_{N-1}$, $g^{(N-1)}_{a_{N-1} b_{N-1}}$)  
and construct an $n$-dimensional Lorentzian manifold $M$  as a direct product of $M_0$, $M_1$, \dots, $M_{N-1}$. For definiteness, let  us assume that  $M_0$ is Lorentzian and $ M_1 \dots M_{N-1} $ are Riemannian.  
 Corresponding dimensions and the Ricci scalars will be denoted by $n_\alpha$ and $R_\alpha$, respectively ($\alpha=0 \dots N-1$). 

All tensors $T$ that can be split similarly as the product metric (i.e. all mixed components vanish and 
 $T_{a_\alpha b_\alpha \dots c_\alpha}=T^{(\alpha)}_{a_\alpha b_\alpha \dots c_\alpha} $ for all values of $\alpha$) are called decomposable. The Ricci and the Riemann tensors are decomposable while the Weyl tensor is not \cite{Ficken39} (see also 
section 4 of \cite{PraPraOrt07}).

Since the Ricci tensor is decomposable it follows that $M$ is Einstein if and only if each $M_\alpha$ is Einstein and 
\be
\frac{R_\alpha}{n_\alpha} = \frac{R_0}{n_0},  \ \ \  \alpha=1\dots N-1.\label{Einstcond}
\ee
From \eqref{QGterm} { and \eqref{RiemannWeyl}, it follows that for universal spacetimes}
\be
R_{acde} R_b^{\ cde} = {\tilde{K}} g_{ab} ,\label{QGtermR}
\ee
where ${\tilde{K}}$ is constant.

Let us now assume that each block  ($M_{\alpha}$, $g^{(\alpha)}_{a_{\alpha} b_{\alpha}}$)  is universal and therefore
\be
R^{(\alpha)}_{a_{\alpha} c_{\alpha}d_{\alpha}e_{\alpha}} R_{\ b_{\alpha}}^{(\alpha)  c_{\alpha}d_{\alpha}e_{\alpha}} 
= \frac{K_\alpha}{n_\alpha}   g^{(\alpha)}_{a_\alpha b_\alpha},  \label{QGcondition}
\ee
where $K_\alpha = R^{(\alpha)}_{acde} R^{(\alpha) acde}$ is the { (constant)} Kretschmann scalar of $M_{\alpha}$.  
{ As a consequence of the decomposability of the Riemann tensor, $R_{acde}  R_{b}^{\ cde}$ is also decomposable (in fact, all rank-2 tensors constructed from the Riemann tensor without covariant derivatives are decomposable).} It thus follows that for $M$
\be
 R_{acde} R_{ b}^{\ cde} = \frac{K}{n}  g_{ab}  \  \Leftrightarrow  \ \frac{K_\alpha}{n_\alpha} = \frac{K_0}{n_0}, \   \alpha=1\dots N-1. \label{MQG} 
\ee 

Now assume that each $M_\alpha$ is maximally symmetric (and therefore also Einstein). With this assumption, each $M_\alpha$ obeys \eqref{QGcondition} with
\be
K_\alpha = \frac{2 {R_\alpha}^{2}}{n_\alpha (n_\alpha-1)}
\ee 
(this is the value of the Kretschmann scalar for maximally symmetric spaces). 

The product manifold $M$ is then Einstein and obeys \eqref{MQG} if and only if (see also \eqref{equalnR})
\be
n_\alpha = n_0, \quad R_\alpha = R_0, \quad  \alpha=1\dots N-1. \label{QGbackgroundcond}
\ee

If each $M_\alpha$ is maximally symmetric, of the same dimension and with the same value of the Ricci scalar, it is clear that any polynomial curvature invariant constructed from the Riemann tensor and its covariant derivatives\footnote{In this case, covariant derivatives of the Riemann tensor   obviously vanish.} has the same constant value for all  $M_\alpha$. Since each $M_\alpha$ is universal, for each $M_\alpha$, all rank-2 tensors constructed from the Riemann tensor and its derivatives are proportional to the metric. As argued above,
the constant of proportionality (being a curvature invariant) is (for a given rank-2 tensor) the same for all $M_\alpha$, $\alpha = 0 \dots N-1 $, and  thus also all rank-2 tensors constructed from the Riemann tensor on $M$ are proportional to the metric on $M$. 
 Note that the above statement also holds for the Weyl tensor (this can be shown using \eqref{RiemannWeyl} and mathematical induction).
Thus we arrive at
\begin{prop}
\label{prop_direct_universal}
Let $M$= $M_0 \times M_1 \times \dots \times M_{N-1}$ and  let all $M_\alpha$, $\alpha=0 \dots N-1$, be non-flat maximally symmetric spaces.
$M$ is universal if and only if the dimensions and the Ricci scalars of each block  $M_\alpha$ coincide 
(i.e.\ $n_\alpha=n_0$,  $R_\alpha=R_0$ for all values of $\alpha$).\footnote{Note that this proposition in fact holds for arbitrary signatures of $M_\alpha$. } 
\end{prop}

In contrast with rank-2 tensors constructed from the Weyl tensor, the Weyl tensor itself is not decomposable and 
although each $M_\alpha$ is conformally flat, $M$ is of the Weyl type D. Corresponding  frame components of the Weyl  tensor  are  given in appendix \ref{Sec_Weylprod}. 
Note that for $n_0=2$, the mWAND $\bl$ is recurrent while for $n_0>2$, it is not.
 
According to the above proposition, one can construct a  type D universal spacetime as a direct product of maximally symmetric spaces if and only if the dimension of $M$ is a composite number. For prime number dimensions, such universal spacetimes obviously do not exist.

\subsection{Type II Kundt universal metrics}
\label{subsecKundtunivII}

Now, let us  construct more general universal metrics by replacing the maximally symmetric Lorentzian space $M_0$ from proposition \ref{prop_direct_universal} by an appropriate  Kundt spacetime.

\begin{prop}
\label{prop_direct_universal2}
Let $M =  M_0 \times M_1 \times \dots \times M_{N-1}$, where $M_0$ is a Lorentzian manifold and $M_1 \dots M_{N-1}$ are non-flat Riemannian maximally symmetric spaces. Let all blocks $ M_\alpha, \alpha=0 \dots N-1$, be of the same dimension and with the same value of the Ricci scalar $R_\alpha$.  If 
$M_0$ is a Kundt proper Einstein ($\Lambda \not=0$) genuine type III or N universal spacetime then $M$ is a type II universal spacetime.   
\end{prop}

\begin{proof}
By Propositions 4.1 and 5.1 of \cite{universal1}, covariant derivatives of the Weyl tensor for a type N and III Einstein Kundt spacetime have only negative b.w. components.
Hence
\begin{lem}
\label{lem_invarsAdS}
Curvature invariants of type N and III Einstein Kundt spacetimes are identical to those of (A)dS with the same value of $\Lambda$.
\end{lem}
Let us take a conserved symmetric rank-2 tensor $\bT$ on $M$ constructed from the Riemann tensor  and its covariant derivatives.
From the decomposability of the Riemann tensor and  its covariant derivatives, it follows that $\bT$ is in general either decomposable or a tensor product of two vectors.   
Note that all vectors constructed from the Riemann tensor and its covariant derivatives
must contain at least one odd derivative of the Riemann tensor. Since in our case covariant derivatives of the Riemann tensors of $M_1,  \dots ,M_{N-1}$ vanish the tensor product  $\bT$ is also decomposable.

 It thus follows that $\bT={\rm diag}(\bT^{(0)},\bT^{(1)}, \dots , \bT^{(N-1)}) $, where $\bT^{(\alpha)}, \ \alpha=0 \dots N-1$, 
is an analog of $\bT$ on $M_{(\alpha)}$.  
 Note that by construction $\bT^{(1)}, \dots, \bT^{(N-1)}$ are conserved and thus (due to the conservation of $\bT$) $\bT^{(0)}$ is also conserved.

By assumptions of proposition \ref{prop_direct_universal2}, all spaces $M_{(\alpha)}, \ \alpha=0 \dots N-1$, are universal and thus it follows that  
$\bT^{(\alpha)} =\lambda^{(\alpha)} \bg^{(\alpha)} $ for all values of $\alpha$. 
Furthermore, $\lambda^{(\alpha)}$ is a curvature invariant on $M_{(\alpha)}$ (proportional to the trace of $T^{(\alpha)}$
and not affected by negative b.w.\ components) and thus by lemma 
\ref{lem_invarsAdS},   $\lambda^{(0)} = \lambda^{(1)} = \dots = \lambda^{(N-1)}$. 
Consequently $\bT$ is proportional to the metric $\bg={\rm diag}(\bg^{(0)},\bg^{(1)}, \dots , \bg^{(N-1)}) $ 
on $M$ and since $\bT$ is an arbitrary conserved rank-2 tensor  $M$ is a universal spacetime.
\end{proof}
In section 6.2 of \cite{universal1}, various explicit examples of type N  universal metrics are given that can be used
to construct type II universal metrics using  Proposition \ref{prop_direct_universal2}.
In contrast,  all known type III  universal spacetimes  are Ricci-flat
and thus at present, they cannot be used in this way. 

Note that in proposition \ref{prop_direct_universal2},  the condition that $M_0$ is of the genuine type  III or N implies that the dimension of
$M_0$ (and thus also of all other blocks) is at least 4. 
The unique mWAND $\bl$ of $M_0$ corresponds to the double WAND of $M$.
By comparing components of $\Phi_{ij}$ expressed in \eqref{KundtRNV} and  \eqref{PhiIJprodk}, it follows that 
 $\bl$ is not recurrent $(\tau_i \not=0)$.

\section{Higher-dimensional generalizations of the Khlebnikov--Ghanam--Thompson metric}
\label{sec_GTmain}

In section 5.2 of \cite{Coleyetal08}, a higher-dimensional generalization of the Khlebnikov--Ghanam--Thompson metric \cite{Khleb,GT2001,GibbonsPope2008}, 
representing gravitational radiation in the (anti-)Nariai spacetime \cite{Ort2002,OrtPod2003}, was studied 
consisting of two blocks of dimensions 2 and $n-2$, where the $n-2$-dimensional space was considered to be maximally symmetric and it was observed
that while the four-dimensional case is universal, the higher-dimensional generalization is not. 

Here, we present 
 different  higher-dimensional  generalizations of the Khlebnikov--Ghanam--Thompson metric consisting of $N$ 2-blocks\footnote{The 2-block case
is in fact a special vacuum subcase of metrics discussed in \cite{KPZK2012}, see eqs. (2.1) and (6.25) there.}
 (section \ref{sec_GT}) or 3-blocks
 (section \ref{sec_GT_3}) with all Riemannian blocks
being maximally symmetric.
We study conserved symmetric rank-2 tensors constructed  from the Riemann tensor and we find that for 2-blocks, they are proportional to the metric
while for 3-blocks, an additional condition arises. For 2-blocks, we also show that all conserved symmetric rank-2 tensors 
constructed  from the Riemann tensor and its derivatives up to the second order are also proportional to the metric. Thus these metrics are 2-universal and solve
vacuum equations of e.g.\ all $L$(Riemann) gravities.
 We also conjecture that a certain subclass of these spacetimes is universal.

\subsection{Type II Kundt spacetimes with boost order -2 covariant derivatives of the Riemann tensor}
\label{sec_bwm2derivatives}

In this section, we prove that under certain assumptions 
(that are satisfied for the metrics given in sections \ref{sec_GT} and \ref{sec_GT_3}) all covariant derivatives of the Riemann tensor 
of a type II Einstein Kundt spacetime are at most of boost order $-2$.

\begin{prop}
\label{lemma_derbalanced}
For a type II Einstein Kundt spacetime admitting a null frame parallelly propagated along an mWAND $\bl$, for which the following assumptions are satisfied
\begin{enumerate}
\item
$\Psi'_{ijk}=0$,
\item
$D \Omega'_{ij}=0$,
\item
the boost order of $\nabla^{(1)} C$ is at most $-2$,
\end{enumerate}
all covariant derivatives of the Weyl tensor $\nabla^{(k)} C$, $k \geq 1$, are at most of boost order $-2$.
\end{prop}

\begin{proof}
Under the assumptions of the proposition, the Ricci and Bianchi equations for the quantities of b.w.\footnote{In the balanced scalar approach, we are interested in the b.w. under boosts with {\em constant} $\lambda$, in this sense e.g.\ $L_{11}$ has b.w.\ $-1$.  } $b$ imply
\BEA
&& b=-2:  \ \ D^3 \kappa'_i 
= 0, \ \  D {\Omega'}_{ij} =0,  \label{eq1balanced1} \\
&& b=-1:  \ \ D^2 L_{11}=0,  \ \  D^2 \rho'_{ij}=0,\  \ D^2  \M{i}{j}{1} = 0, \\
&& b=0:  \ \ \ \  D {\Phi_{ij}} = 0 , \ \  D {\Phi_{ijkl}} = 0 , \ \  D \M{i}{j}{k} =0,\ \ D\tau_i=0, \ \ DL_{1i}=0. \label{eq1balanced3}
\EEA  

Following \cite{Coleyetal04vsi} and \cite{universal1}, we define a 1-balanced scalar and a 1-balanced tensor: Let us  say that a  scalar $\eta$ with the b.w. $b$  is 1-balanced if $D^{-b-1} \eta =0$ for $b<-1$ and $\eta=0$ for $b \geq -1$ and that a tensor is 1-balanced if  all its frame  components are 1-balanced scalars. Obviously a
1-balanced tensor admits only non-vanishing components of b.w.\ $\leq -2$.

Since the conditions \eqref{eq1balanced1}-\eqref{eq1balanced3} are the same as eqs.\ (28)--(30) for type N spacetimes 
in \cite{universal1} (except for the conditions for b.w.\ 0 components that are missing in type N), one can use the same proof to show that for a 1-balanced scalar $\eta$, scalars $L_{1i} \eta$, $\tau_{i} \eta$, $L_{11} \eta$, $ \kappa'_i 
\eta $,  $ \rho'_{ij} 
\eta$,  
$\Mi_{\!j1}\!  \eta$ and $ \Mi_{\!kl}\!  \eta$ and 
$D \eta ,\ \delta_i \eta,\ \T \eta $ are  1-balanced scalars.
 Note that if  we denote frame components of a tensor by $\eta_i$ then frame components of its covariant derivative consist of terms 
$L_{1i} \eta_i$, $\tau_{i} \eta_i$, $\dots $, 
$\delta_i \eta,\ \T \eta $.
 
It thus follows that 

\begin{lem}
\label{lemma1balanced}
For type II Einstein Kundt spacetimes obeying \eqref{eq1balanced1}--\eqref{eq1balanced3} in a frame parallelly propagated along an mWAND $\bl$, 
 a covariant derivative of a 1-balanced tensor is a 1-balanced tensor. 
\end{lem}

The Weyl tensor for type II Einstein Kundt spacetimes is {\em not} 1-balanced, however, 
 in particular cases (including the examples discussed in sections \ref{sec_GT}, \ref{sec_GT_3}),  $\nabla^{(1)} C$ is 1-balanced. 
Then, by lemma \ref{lemma1balanced}, proposition \ref{lemma_derbalanced} follows.

\end{proof}

Let us conclude this section with the following lemma
\begin{lem}
\label{lemmaconserved}
 For spacetimes obeying the assumptions of proposition \ref{lemma_derbalanced},  all rank-2 tensors $T_{ab}$ constructed from the Riemann
 tensor and its covariant derivatives obey (for $k \geq 1$): \\
(i) If $T_{ab}$ is quadratic or of higher order in $\nabla^{(k)} C$ then it vanishes. \\
(ii) If $T_{ab}$ is linear in  $\nabla^{(k)} C$ then it is conserved. 
 \end{lem}

\begin{proof}
(i) All components of an arbitrary rank-2 tensor $T_{ab}$ have b.w.\ $\geq -2$, however, 
{ for a tensor obeying the assumptions of proposition \ref{lemma_derbalanced}}, terms quadratic or of higher order in $\nabla^{(k)} C$ have 
b.w.\ $\leq -4$.\\
(ii) By proposition \ref{lemma_derbalanced}, 
$T_{ab;c}$ only admits  components of b.w.\ $\leq -2$ and therefore its trace vanishes since a rank-1 tensor has components
of b.w.\ $\geq -1$.
\end{proof}


\subsection{A generalization of the Khlebnikov--Ghanam--Thompson metric consisting of 2-blocks}

\label{sec_GT}

 Here, let us study necessary conditions 
for universality of
a higher-dimensional  generalization of the Khlebnikov--Ghanam--Thompson metric consisting of $N$ 2-blocks
\begin{equation}
  \d s^2 = 2 \d u \d v + (\lambda v^2 + H(u, x_\alpha, y_\alpha)) \d u^2 + \frac{1}{|\lambda|} \sum_{\alpha=1}^{N-1} (\d x_\alpha^2 + \s^2(x_\alpha) \, \d y_\alpha^2), \ \ \ 
	\label{GT-2blocks}
\end{equation}
where $\s(x_\alpha) = \sin (x_\alpha)$ for $\lambda >0$,  $\s(x_\alpha) = \sinh (x_\alpha)$ for $\lambda <0$. The metric \eqref{GT-2blocks} is Einstein iff $H$ obeys 
\beq
 \Box H &=&\left[\sum_{\alpha=0}^{N-1}  \Box^{(\alpha)} \right] H=0, \  \ \ 
\Box^{(\alpha)}\equiv\nabla^{a_{(\alpha)}}\nabla_{a_{(\alpha)}},\nonumber\\
 \Box^{(\alpha)}H&=& |\lambda| \left( H_{,x_\alpha x_\alpha} + \frac{1}{\s^2(x_\alpha)} H_{,y_\alpha y_\alpha} 
+ \ct(x_\alpha) H_{,x_\alpha}\right)\ \mbox{for}\ \alpha=1,\dots, N-1,
\ \ \  
\label{GT-BoxH}
\eeq
where $\cs(x_\alpha)=\cos (x_\alpha)$, $\ct(x_\alpha)=\frac{\cs(x_\alpha)}{\s(x_\alpha)}$ for $\lambda >0$,  $\cs(x_\alpha) = \cosh (x_\alpha)$ for $\lambda <0$
and  where 
$\Box^{(0)} H$ vanishes identically.
 Note that 
in contrast with the cases discussed in section \ref{sec_examples}, 
the Riemann and Ricci tensors are not decomposable,
 due to the dependence of $H$ on  $x_\alpha$ and $y_\alpha$, and non-zero negative b.w.\ components of the Riemann and Ricci tensors
appear (see  appendix \ref{app_GT}).

These type II Einstein Kundt  metrics admit a recurrent $(\tau_i=0)$ multiple WAND $\bl=\d u$. In a frame 
parallelly propagated along $\bl$, in which
$\Psi'_{ijk}=0$ (this follows from $\rho'_{ij}=0$ and the Ricci eq.\ (11l) in \cite{OrtPraPra07}) and $D\Omega'_{ij}=0$, 
the boost order of $\nabla^{(1)} C$ is at most $-2$ (see appendix \ref{app_GT}).
 Thus, all the assumptions of proposition \ref{lemma_derbalanced} are satisfied 
and therefore all covariant derivatives of the Weyl tensor admit only terms with b.w. $-2$ or less.

The b.w.\ zero part of the Riemann tensor is decomposable and can be written as
\beq
\Rabcd \!\!=\!\!\Lambda\!\left[(\delta^{(0)})^a_{~c}(\delta^{(0)})^b_{~d}-(\delta^{(0)})^a_{~d}(\delta^{(0)})^b_{~c}+\!\sum_{\alpha=1}^{N-1}
\left((\delta^{(\alpha)})^a_{~c}(\delta^{(\alpha)})^b_{~d}-(\delta^{(\alpha)})^a_{~d}(\delta^{(\alpha)})^b_{~c}\right)\right]\!, \Lambda=\frac{\lambda}{n_0-1}.
\label{RiemannGT}\eeq
Note that in an appropriately chosen frame (e.g.\ the frame \eqref{frame-2} given in appendix \ref{app_GT}), the b.w.\ zero frame components of the Riemann
tensor are also decomposable.

The operator $\delta^{(\alpha)}:TM\rightarrow TM$ can be seen as a projection operator projecting onto the $\alpha^{\mbox{th}}$ piece. They fulfill (as operators): 
\[ \delta^{(\alpha)}\delta^{(\beta)}=\begin{cases} 0, & \alpha\neq \beta \\ \delta^{(\alpha)}, 
& \alpha=\beta\end{cases}, \quad {\rm Id}_{TM}=\bigoplus_{\alpha=0}^{N-1}\delta^{(\alpha)}. \]  
We also note that the b.w.\ 0 component of the Riemann tensor possesses the discrete symmetries:
\be
 \sigma(\alpha,\beta): ~~\delta^{(\alpha)} \leftrightarrow \delta^{(\beta)}.  \label{permut}
\ee
These $\sigma(\alpha,\beta)$'s generate the group of permutations of the $N$ projection operators. 

\subsubsection{Rank-2 tensors constructed from the Riemann tensor}
\label{sec_rank2riemann}

In this section, we show that all symmetric rank-2 tensors $T_{ab}$ constructed from the Riemann tensor are proportional to the metric,
i.e.\ this metric is 0-universal.

As mentioned above, the b.w.\ 0 part of the Riemann tensor \eqref{RiemannGT} is decomposable. Therefore, any symmetric rank-2 tensor $T_{ab}$ constructed from $\Rabcd$ is also decomposable and due to \eqref{RiemannGT} also proportional to the metric. 

Let us proceed with terms containing b.w.\ $-2$ components of the Weyl tensor $\C_{abcd}$ 
(note that for Einstein spacetimes, these are equal to b.w.\ $-2$ components of the Riemann tensor). Obviously, a non-vanishing rank-2 tensor constructed from $\Rabcd $ and  $\C_{abcd}$ 
can contain at most one term $\C_{abcd}$.  
Therefore, we are left with the cases
\beq
{\mathcal{R}}^{ab} \ \C_{\bullet  a  \bullet b},  \label{SC1} \\
{\mathcal{R}}^{abc}_{\phantom{abc} \bullet}\  \C_{   abc  \bullet}, \label{SC2} \\
{\mathcal{R}}^{abcd}_{\phantom{abcd} \bullet \bullet}\  \C_{ abcd}, \label{SC3}
\eeq
where ${\mathcal{R}}^{\dots}$ are rank-2, 4 and 6 tensors constructed from $\Rabcd$ and free indices are indicated by $\bullet$. 

\begin{itemize}
\item ${\mathcal{R}}^{ab} \ \C_{\bullet  a  \bullet b}$:
Since ${\mathcal{R}}^{ab}$ is proportional to the metric and $\C_{ abcd}$ is traceless, the case \eqref{SC1} obviously vanishes.
\item
${\mathcal{R}}^{abc}_{\phantom{abc} \bullet}\  \C_{   abc  \bullet}$:
Since $\Rabcd$ 
is decomposable it follows that ${\mathcal{R}}^{abcd}$ is either decomposable or a tensor product of two rank-2 decomposable tensors 
(which are both proportional to the metric). 
\begin{itemize}
\item
In the latter ``$2 \otimes 2$'' case,  \eqref{SC2} obviously vanishes. 
\item
Let us proceed with the rank-4 decomposable case:\\
The tensor \eqref{SC2} has b.w.\ $-2$ and therefore it can be expressed as
\be
{\mathcal{R}}^{abc}_{\phantom{abc} 1} \C_{ abc  1} \ell_{\bullet} \ell_{\bullet}\ ,\label{case2}
\ee
where indices $a, b, c$ and $1$ are now understood as frame indices. $\C_{abc  1}$ can be non-vanishing only for $c$ taking value $i$ (i.e.\ belonging to the spacelike blocks). However, in such a case,    
\eqref{case2} and \eqref{SC2} vanish due to the decomposability of the first term. 
\end{itemize}
\item
${\mathcal{R}}^{abcd}_{\phantom{abcd} \bullet \bullet}\  \C_{ abcd}$:
Similarly as above, ${\mathcal{R}}^{abcdef}$  appearing in \eqref{SC3} is decomposable or it is a $2 \otimes 2 \otimes 2$ or $2 \otimes 4$ tensor product of decomposable tensors.
\begin{itemize}
\item
If ${\mathcal{R}}^{abcdef}$   is decomposable  then \eqref{SC3} vanishes (one can use similar arguments as for the vanishing of \eqref{case2}).
\item
For the case $2 \otimes 2 \otimes 2$, \eqref{SC3} vanishes due to the tracelessness of $\C_{abcd}$.
\item
In the $2 \otimes 4$ case, ${\mathcal{R}}^{abcdef}={\mathcal{R}}^{ab} {\mathcal{R}}^{cdef}$. If ${\mathcal{R}}^{ab}$ has no free index, then \eqref{SC3} vanishes due to
 the tracelessness of $\C_{abcd}$. If ${\mathcal{R}}^{ab}$ has one free index then  \eqref{SC3} 
 vanishes  thanks to the decomposability of ${\mathcal{R}}^{cdef}$, i.e. terms like ${\mathcal{R}}^{0i1j}$ vanish.
If ${\mathcal{R}}^{ab}$ has two free indices, then \eqref{SC3} vanishes due to the vanishing of $g_{11}$. 
\end{itemize}
\end{itemize}
We thus arrive at

\begin{prop}
\label{prop_rank2Riem}
The metric \eqref{GT-2blocks}, \eqref{GT-BoxH} 
is 0-universal.
\end{prop}

Obviously, it follows that
\begin{cor}
The metric \eqref{GT-2blocks}, \eqref{GT-BoxH} is a vacuum solution to all gravitational theories  with field equations (derived from the Lagrangian \eqref{Lagr})
that may contain derivatives of the Ricci tensor but do not contain derivatives of the Riemann tensor.
\end{cor}
Examples of such theories are Lovelock gravity (no derivatives of the Ricci and the Riemann tensors) or quadratic gravity (the field equations
of this theory contain derivatives
of the Ricci tensor but do not contain derivatives of the Riemann tensor).


\subsubsection{Rank-2 tensors constructed from the Riemann tensor and its 2nd covariant derivatives}
\label{sec_rank2-2ndorder}

Recall that by proposition \ref{lemma_derbalanced}  for the metrics \eqref{GT-2blocks}, \eqref{GT-BoxH}, all covariant derivatives of the Riemann tensor are of the boost 
order $\leq -2$ and thus  a non-vanishing rank-2 tensor constructed from the Riemann tensor and its derivatives may contain at most one term with a derivative.
Therefore, it is sufficient to study only terms which are linear in (obviously even) covariant derivatives of the Riemann tensor 
that are conserved by  lemma \ref{lemmaconserved}. 
Thus, let us  proceed with studying rank-2 tensors constructed from the Riemann tensor and its second covariant derivatives.

The requirement that conserved symmetric rank-2 tensors containing second covariant derivatives of the Riemann tensor are proportional to the metric
leads to a new  necessary condition on the metric function $H$. For instance,  the following rank-2 tensor 
\be
{R^{cg}}_{eh}{R^{dh}}_{fg}\nabla^{(e}\nabla^{f)}C_{acbd}  \label{GT-RRnablaC}
\ee
vanishes for the metric \eqref{GT-2blocks}, \eqref{GT-BoxH} iff 
\be
\left[\sum_{\alpha=0}^{N-1}  (\Box^{(\alpha)})^2 \right] H=0. \label{GT-sumBox2}
\ee

In the following, we show that  \eqref{GT-sumBox2} is in fact a sufficient 
condition on $H$ for 2-universality of the metrics \eqref{GT-2blocks}, \eqref{GT-BoxH}.

Let us start with the following lemma:
\begin{lem}
\label{lem_ddC_tracefree}
For the metrics \eqref{GT-2blocks}, \eqref{GT-BoxH}, all rank-2 tensors constructed from the Riemann tensor and its 2nd covariant derivatives that contain any trace of $\nabla_a \nabla_b C_{cdef}$ vanish.
\end{lem}
\begin{proof}
Since the Weyl tensor is traceless, we only have to consider terms i) $\nabla_a \nabla^c C_{cdef}$, ii) $\nabla^c \nabla_b C_{cdef}$ and iii) $\nabla^b \nabla_b C_{cdef}$. Due to the Bianchi identity, the case i) vanishes identically. The case ii) can be expressed in terms of the case i) using the expression for the
commutator of covariant derivatives for any tensor $T$
\BE
[\nabla_a,\nabla_b]T_{c_1.\dots c_k}=T_{d\dots c_k}R^d_{\ c_1 ab}+\dots +T_{c_1\dots d}R^d_{\ c_kab}. \label{commut_der}
\EE
By proposition \ref{lemma_derbalanced} for $T$ being the Weyl tensor, the left-hand side of \eqref{commut_der} is of the boost order $\leq -2$. Therefore, the right-hand side of \eqref{commut_der}
is also of the boost order $\leq -2$. The corresponding rank-2 contraction with further Riemann terms is of the boost order $\leq -2$, however, taking into account proposition 
\ref{prop_rank2Riem},
 such contraction has to vanish. The case iii) can be easily expressed in terms of case ii) terms using the Bianchi identity.
\end{proof}

Note that  for the  metrics \eqref{GT-2blocks}, \eqref{GT-BoxH}, the contraction $C_{abcd;e} \ell^a \ell^c $ vanishes. By further differentiation (and taking into account that $\bl$ is recurrent), one arrives at
the useful relation
\be
C_{abcd;ef} \ell^a \ell^c =0. \label{eqddCll}
\ee

Now, let us  proceed with possible rank-2 tensors constructed from the Riemann tensor  and its 2nd covariant derivatives. 
Cases to consider are
\begin{enumerate}
\item
$
{\mathcal{R}}^{abcd}  \nabla_{(a} \nabla_{b)} C_{c \bullet d \bullet}\ ,
$
\item
$
{\mathcal{R}}^{\bullet abcde}  \nabla_{(a} \nabla_{b)} C_{\bullet cde}\ ,
$
\item
$
{\mathcal{R}}^{\bullet \bullet abcdef}  \nabla_{(a} \nabla_{b)} C_{cdef}\ .
$ 
\end{enumerate}
 We can consider only the symmetric parts in the cases 1, 2, and 3 since from \eqref{commut_der}, the antisymmetric parts are proportional to the  terms treated in section \ref{sec_rank2riemann} that are proportional to the metric.
Note that  all cases with one or two free indices in the $\nabla_{a} \nabla_{b} C_{cdef}$ part can be reduced to the case 2 and 1, respectively. This can be
shown using the Bianchi identities, a relation for the commutator of covariant derivatives \eqref{commut_der} and  
 the results of  section \ref{sec_rank2riemann}.\footnote{ 
 1A) Terms $\nabla_{.} \nabla_{\bullet} C_{ \bullet \dots}$  can be converted to terms belonging to the case 1 using the Bianchi identities; 
1B) terms $\nabla_{\bullet} \nabla_{.} C_{ \bullet \dots}$ (using the commutator \eqref{commut_der})  consist of terms 1A) plus terms not involving covariant derivatives, discussed in section \ref{sec_rank2riemann}, that are proportional to the metric;
1C)  terms $\nabla_{\bullet} \nabla_{\bullet} C_{ ....}$ can be expressed as  terms 1B) using the Bianchi identities.
Similarly, cases with only one free index in the $\nabla_{a} \nabla_{b} C_{cdef}$ part can be converted to the case 2.  }

\begin{itemize}
\item
${\mathcal{R}}^{abcd}  \nabla_{(a} \nabla_{b)} C_{c \bullet d \bullet}$:\\
${\mathcal{R}}^{abcd}$ is either $2 \otimes 2  $ or decomposable. 
\begin{itemize}
\item
The $2 \otimes 2  $ case is trivial due to 
lemma \ref{lem_ddC_tracefree}.
\item
The decomposable case: since the tensor product commutes with permutations $\sigma(\alpha,\beta)$,  ${\mathcal{R}}^{ab}_{\ \ cd}$ 
has to preserve permutations as well. It is thus a symmetric polynomial in $\delta^{(\alpha)}$'s.
${\mathcal{R}}^{ab}_{\ \ cd}$ is therefore a linear combination of $\sum_{\alpha=0}^{N-1}
(\delta^{(\alpha)})^a_{~c}(\delta^{(\alpha)})^b_{~d}$ and $\sum_{\alpha=0}^{N-1}
(\delta^{(\alpha)})^a_{~d}(\delta^{(\alpha)})^b_{~c}$, however, due to the symmetries of the expression 1, 
we can consider just the first term that leads to \eqref{GT-sumBox2} (see also \eqref{RR-D}).
\end{itemize}
The case 1 is thus either trivial or it reduces to the  condition \eqref{GT-sumBox2} following from \eqref{GT-RRnablaC}.
\item
${\mathcal{R}}^{\bullet abcde}  \nabla_{(a} \nabla_{b)} C_{\bullet cde}$:\\
${\mathcal{R}}^{abcdef}$ belongs to one of the following subcases:\\
i) $2 \otimes 2 \otimes 2 $ \\
ii) $2  \otimes 4$ \\
iii) decomposable
\begin{itemize}
\item
In the case i), the expression 2 obviously vanishes due to lemma \ref{lem_ddC_tracefree}. 
\item
 In the case ii),  
${\mathcal{R}}^{abcdef}={\mathcal{R}}^{ab}{\mathcal{R}}^{cdef}$. The free index now belongs either to ${\mathcal{R}}^{ab}$ or to 
${\mathcal{R}}^{cdef}$. The first case reduces to the case 1, while the second case vanishes by lemma \ref{lem_ddC_tracefree}.
\item
The case iii):  Since the expression 2 has b.w.\ $-2$, both free indices correspond to b.w.\ $-1$ (i.e.\ the upper index is 0 or the lower one is 1). Since  
${\mathcal{R}}^{abcdef}$ is decomposable, all indices in the expression 2 are either 1 or 0. Taking into account that the first term 
${\mathcal{R}}^{abcdef}$ is of b.w.\ 0, while the second term $\nabla_{(a} \nabla_{b)} C_{\bullet cde}$ is of b.w.\ $-2$, we find that the expression 2
reduces to ${\mathcal{R}}^{010101} \nabla_{(1} \nabla_{1)} C_{1010}$. The second term, 
$\nabla_{(1} \nabla_{1)} C_{1010}=\nabla_{(a} \nabla_{b)} C_{cdef} \ell^c n^d \ell^e n^f n^a n^b$, vanishes due
to \eqref{eqddCll}. 
\end{itemize}
Case 2 terms thus do not lead to  new conditions on the metric.
\item
${\mathcal{R}}^{\bullet \bullet abcdef}  \nabla_{(a} \nabla_{b)} C_{cdef}$:\\ 
${\mathcal{R}}^{abcdefgh}$ belongs to one of the following subcases:\\
i) $2 \otimes 2 \otimes 2 \otimes 2$ \\
ii) $2 \otimes 2 \otimes 4$ \\
iii) $ 2 \otimes 6$ \\
iv) decomposable 
\begin{itemize}
\item
In the case i), the expression 3 obviously vanishes due to lemma \ref{lem_ddC_tracefree}.  
\item
In the case ii), there is either a rank-2 tensor ${\mathcal{R}}^{ab}$ with both dummy indices, which vanishes by lemma \ref{lem_ddC_tracefree}, or ${\mathcal{R}}^{ab}$ has one free index, which reduces to the   case 2 ii).
\item
If in the case iii), ${\mathcal{R}}^{ab}$ has no free index or one free index then again the expression 3 vanishes or reduces to the case 2 iii), respectively.
If ${\mathcal{R}}^{ab}$ has two free indices then the expression 3 is ${\mathcal{R}}^{\bullet \bullet}$ times a full contraction of a b.w.\ $-2$ tensor, which is obviously zero. Note that ${\mathcal{R}}_{11}$ vanishes since $g_{11}=0$. 
\item
The case iv) is similar to the case 2 iii) and similarly we obtain a product of a b.w.\ zero component ${\mathcal{R}}^{abcdefgh}$ with indices either 1 or 0
and $\nabla_{1} \nabla_{1} C_{1010}$ that vanishes (see \eqref{eqddCll}).
\end{itemize}
Thus case 3 terms also do not lead to  new conditions on the metric.
\end{itemize}

We can thus conclude with 
\begin{prop}
\label{prop_rank2Riem_der}
The metric \eqref{GT-2blocks}, \eqref{GT-BoxH}, obeying \eqref{GT-sumBox2} is 2-universal.
\end{prop}

Thus in addition to the theories mentioned at the end of  section \ref{sec_rank2riemann}, the metric \eqref{GT-2blocks}, \eqref{GT-BoxH} obeying \eqref{GT-sumBox2}
 also solves the vacuum equations of e.g.\ all $L($Riemann$)$ gravities.


\subsubsection{Rank-2 tensors constructed from the Riemann tensor and its 4th and higher covariant derivatives}

Now, let us  proceed with studying rank-2 tensors constructed from the Riemann tensor and its covariant derivatives of the 4th and higher order.
As in the section \ref{sec_rank2-2ndorder} , it is sufficient to study only terms which are linear in  covariant derivatives of the Riemann tensor 
that are conserved by  lemma \ref{lemmaconserved}. 

For higher-order derivatives, we will have more possibilities of $\mathcal{R}^{ab...cd}$ constructed from tensor products of the b.w.\ 0 components of the Riemann tensor, eq.\ (\ref{RiemannGT}). Let us consider a general case where $N$ blocks are of a dimension $n_0$. Then, using eq.\ (\ref{RiemannGT})  as well as the permutation symmetry generated by eq.\ (\ref{permut}), the tensor  $\mathcal{R}^{ab...cd}$ should be invariant under this permutation symmetry. A classical result from invariant theory states that invariant polynomials in $N$ variables are generated by the $N$ power sum symmetric polynomials  $s_k(x_1,...,x_N)=(x_1)^k+...+(x_N)^k$, $k=1...N$. We can use this to construct tensors invariant under the permutation symmetry (\ref{permut})
\be
D_k=\left(\delta^{(0)}\right)^{\otimes k}\oplus ...\oplus \left(\delta^{(N-1)}\right)^{\otimes k},
\ee
or in the component form: 
\be
{D_k}^{a_1...a_k}_{\phantom{b_1...b_k}b_1...b_k}\!=\! \sum_{\alpha=0}^{N-1}\left(\delta^{(\alpha)}\right)^{a_1}_{~b_1}\cdots \left(\delta^{(\alpha)}\right)^{a_k}_{~b_k}
\!=\! \left(\delta^{(0)}\right)^{a_1}_{~b_1}\cdots \left(\delta^{(0)}\right)^{a_k}_{~b_k}+...+\left(\delta^{(N-1)}\right)^{a_1}_{~b_1}\cdots \left(\delta^{(N-1)}\right)^{a_k}_{~b_k}.
\ee
To see that these tensors are actually constructable from eq.\ (\ref{RiemannGT}), we first note that the trace of eq.\ (\ref{RiemannGT}) gives the Ricci tensor (there is no loss of generality to set $\Lambda=1$)
\be
R^a_{~b}={D_1}^{a}_{\phantom{b}b}.
\ee
Next, define ${\Delta^{(\alpha)}}^{ab}_{~~cd}=2\left(\delta^{(\alpha)}\right)^{a}_{~[c} \left(\delta^{(\alpha)}\right)^{b}_{~d]}$. Then
\[ {\Delta^{(\alpha)}}^{ab}_{~~cd}{\Delta^{(\beta)}}^{ce}_{~~fg}=0 ~~(\alpha\neq \beta),\] 
while if $\alpha=\beta$ and doing the double contraction:
 \[ {\Delta^{(\alpha)}}^{ab}_{~~cd}{\Delta^{(\alpha)}}^{ce}_{~~ag}=(n_0-1)\left(\delta^{(\alpha)}\right)^{b}_{~d} 
\left(\delta^{(\alpha)}\right)^{e}_{~g}-{\Delta^{(\alpha)}}^{be}_{~~dg}. \]
Therefore,  
\beq
\Rabcd\Rceag&=&\left(\sum_{\alpha=0}^{N-1}{\Delta^{(\alpha)}}^{ab}_{~~cd}\right)
\left(\sum_{\beta=0}^{N-1}{\Delta^{(\beta)}}^{ce}_{~~ag}\right)
=\left(\sum_{\alpha=0}^{N-1}{\Delta^{(\alpha)}}^{ab}_{~~cd}{\Delta^{(\alpha)}}^{ce}_{~~ag}\right)\nonumber \\
&=& (n_0-1)\sum_{\alpha=0}^{N-1}\left(\delta^{(\alpha)}\right)^{b}_{~d} 
\left(\delta^{(\alpha)}\right)^{e}_{~g}-\sum_{\alpha=0}^{N-1}{\Delta^{(\alpha)}}^{be}_{~~dg}=(n_0-1){D_2}^{be}_{\phantom{be}dg}-\Rbedg.
\eeq
Thus
\beq
{D_2}^{be}_{\phantom{be}dg}=\frac{1}{n_0-1}\left(\Rabcd\Rceag+\Rbedg\right).\label{RR-D}
\eeq
Once we have constructed $D_2$, we can constuct the rest of the symmetric tensors $D_k$ iteratively by noting that 
\beq
{D_k}^{a_1...a_k}_{\phantom{b_1...b_k}b_1...b_k}{D_2}^{b_kc}_{\phantom{b_kc}de}={D_{k+1}}^{a_1...a_kc}_{\phantom{b_1...b_ke}b_1...b_{k-1}de}. 
\eeq
 These tensors will give new necessary conditions for the space to be universal (via expressions like \eqref{boxboxC}). 

Let us consider the case where $n_0=2$ and the 4th order derivatives. We proceed similarly as in section \ref{sec_rank2-2ndorder} and it turns out that everything reduces to the following  cases:
\begin{itemize}
\item
The case 1.
$
{\mathcal{R}}^{abcdef}  \nabla_{a} \nabla_{b} \nabla_{c}\nabla_{d}C_{e \bullet f \bullet}
$: The tensor $D_3$ defined above, gives an additional possibility where 
\be
 {D_3}^{ace}_{\phantom{ace}bdf}\nabla_{a} \nabla^{b} \nabla_{c}\nabla^{d}C^{\phantom{e\bullet}f}_{e \bullet ~ \bullet}, \label{boxboxC}
\ee
giving the requirement
\BE
\left[\sum_{\alpha=0}^{N-1} (\Box^{(\alpha)})^3 \right] H=0.
\EE
\item
The case 2.
$
{\mathcal{R}}^{\bullet abcdefg}  \nabla_{a} \nabla_{b}\nabla_{c} \nabla_{d} \ C_{\bullet efg}
$
\begin{itemize}
\item
which for ${\mathcal{R}}^{\bullet abcdefg} $ decomposable, leads to a term proportional to $C_{1010;1011}$
\item
for a tensor product of 2 rank-4 decomposable tensors $4\otimes 4$, leads to  terms like \\
${\mathcal{R}}_{0101}  {\mathcal{R}}^{ijkl}   \ C_{1 i1j;01kl}$
\end{itemize}
\item 
The case 3.
$
{\mathcal{R}}^{\bullet \bullet abcdefgh}  \nabla_{a} \nabla_{b}\nabla_{c}\nabla_{d} C_{efgh}$ - this case is either trivial or reduces to the case 1 or 2.
\end{itemize}

\vspace{5mm}

We note that for higher derivatives, we will get the additional tensors $D_k$ giving additional conditions 
$\left[\sum_{\alpha=0}^{N-1} (\Box^{(\alpha)})^P \right] H=0 $, for some integer $P$. 
In our case, we notice that the symmetric polynomials give rise to symmetric polynomials in the (2nd order differential) operators $\Box^{(\alpha)}$. These operators commute: $[\Box^{(\alpha)},\Box^{(\beta)}]=0$ for the general Khlebnikov--Ghanam--Thompson space. The b.w. 0 component of the Riemann tensor is invariant under the permutation symmetry and thus we expect a rank-2 tensor to be invariant under this permutation symmetry as well. Higher-order derivatives will then give us expressions which are `symmetric polynomials' in the operators $\Box^{(\alpha)}$. Since the operators commute  any symmetric polynomial in the operators $\Box^{(\alpha)}$ can be generated by the $N$ power sum symmetric polynomials $s_k(\Box^{(0)},..., \Box^{(N-1)})$, $k=1,\dots, N$. Consequently, we do not expect any additional conditions then those of the form 
$\left[\sum_{\alpha=0}^{N-1} (\Box^{(\alpha)})^P \right] H=0$. {Indeed, since any symmetric polynomial in the operators $\Box^{(\alpha)}$  is also finitely generated, we can assume $P=1,\dots,N$.} Therefore, we conjecture

\begin{con} \label{con:symmetricCond}
Metrics \eqref{GT-2blocks}, \eqref{GT-BoxH} obeying the following set of $N$ equations
\BE
\left[\sum_{\alpha=0}^{N-1} (\Box^{(\alpha)})^P \right] H=0, \label{GT-sumBoxP}
\EE
where $P=1\dots N$,
are universal.
\end{con}
 We note that all these conditions are necessary for all $P=1,\dots,N$, which can be easily seen by assuming  
$H$ to be a simultaneous eigenvector for all $\Box^{(\alpha)}$: $\Box^{(\alpha)}H=\lambda_{\alpha}H$. 



\subsection{A generalization of the Khlebnikov--Ghanam--Thompson metric consisting of 3-blocks }
\label{sec_GT_3}

Let us study a higher-dimensional generalization of the  Khlebnikov--Ghanam--Thompson metric consisting of  $N$ 3-blocks (see appendix \ref{app_GT_3} for 
the Ricci rotation coefficients and components of the Weyl tensor)
\begin{equation}
  \d s^2 = 2 \d u \d v + H(u, z, x_\alpha, y_\alpha, z_\alpha) \d u^2 + 2 \frac{2v}{z} \d u \d z
  - \frac{2}{\lambda z^2} \d z^2 - \frac{2}{\lambda} \sum_{\alpha=1}^{N-1} \left[\d x_\alpha^2 + sh^2_\alpha \, (\d y_\alpha^2 + s^2_\alpha \d z_\alpha^2)\right],
  \label{appGT3:GTmetric}
\end{equation}
where, in this form, only negative $\lambda$ is allowed and  $s_\alpha = \sin (y_\alpha)$,  $sh_\alpha = \sinh (x_\alpha)$.

The metric \eqref{appGT3:GTmetric}   
is a Kundt metric of the type II, where now,  in contrast with the 2-block case studied in   section \ref{sec_GT},  the mWAND $\bl$ is not recurrent.
 The metric \eqref{appGT3:GTmetric}   
is Einstein iff $H$ obeys
\be
\Box H-2\lambda zH_{,z}=0\label{GT-BoxH_3}
\ee
and then it { satisfies all the assumptions of proposition \ref{lemma_derbalanced} (see appendix \ref{app_GT_3}) and thus }
all covariant derivatives of the Riemann (and Weyl) tensor are of the  boost order $\leq -2$.

Now, let us study conserved symmetric rank-2 tensors constructed from the Riemann tensor. 
The b.w.\ zero part of the Riemann tensor is decomposable as in the 2-block case \eqref{RiemannGT}. 
In order to show that such rank-2 tensors are proportional to the metric, one can essentially  repeat the proof of  section \ref{sec_rank2riemann} except for terms like 
\be
{\mathcal{R}}^{abc}_{\phantom{abc} 1} \C_{ abc  1}, \  
\ee
appearing e.g.\ in \eqref{case2}, that now  do not vanish  due to the decomposability of ${\mathcal{R}}^{abcd}$, since in principle terms like 
${\mathcal{R}}^{212}_{\phantom{abc} 1} \C_{ 2121}={\mathcal{R}}^{212}_{\phantom{abc} 1} {\Omega'}_{22}$ might appear. For the metric 
\eqref{appGT3:GTmetric}, \eqref{GT-BoxH_3},   
${\Omega'}_{22}$ vanishes iff
\be
\sum_{\alpha=1}^{N-1} \Box^{(\alpha)} H=0.\label{GT-sumBox3}
\ee

Rescaling $H$
\be
H=\frac{{\tilde H}}{z^2},
\ee
the Einstein equation reads
\be
\sum_{\alpha=0}^{N-1}  \Box^{(\alpha)} {\tilde H}=0\label{GT-Ein-3}
\ee
and together with \eqref{GT-sumBox3} gives
\be
\Box^{(0)} {\tilde H}=-\frac{\lambda z}{2}\left(z{\tilde H},_{zz}-{\tilde H},_z\right) =0,
\ee
which is satisfied for
\be
{\tilde H}=h_1(u, x_\alpha, y_\alpha, z_\alpha)+h_0 (u,x_\alpha, y_\alpha, z_\alpha) z^2, 
\ \ H=h_0 +\frac{h_1}{z^2}.\label{solH}
\ee

We thus arrive at

\begin{prop}
\label{prop_rank2Riem_GT3}
The metric \eqref{appGT3:GTmetric},  
obeying 
\eqref{GT-sumBox3} and \eqref{solH},
is 0-universal.
\end{prop}

Thus, similarly as in section \ref{sec_GT}, these metrics   solve the vacuum equations
of e.g.\ Lovelock gravity or quadratic gravity.

We expect that a modification of conjecture \ref{con:symmetricCond} is valid also for the metric \eqref{appGT3:GTmetric}.
\begin{con} \label{con:symmetricCond-3block}
The higher-dimensional generalization of the Khlebnikov--Ghanam--Thompson metric constructed from $N$ $3$-dimensional blocks \eqref{appGT3:GTmetric},   
satisfying   
$H,_z=0$  
and  the following set of $N$ equations
\BE
\left[\sum_{\alpha=1}^{N-1} (\Box^{(\alpha)})^P \right] H=0, \label{GT-sumBoxP-gen}
\EE
where $P=1\dots N$, is universal.
\end{con}

Obviously, in a similar way, one could also study a generalization of the Khlebnikov--Ghanam--Thompson metric constructed from 
an appropriate Kundt form of a maximally
symmetric $n_0$-dimensional Lorentzian space and $N-1$ $n_0$-dimensional maximally symmetric Riemannian spaces with $n_0>3$.
Indeed, we expect  a modification of conjecture \ref{con:symmetricCond-3block} 
to be valid for  Khlebnikov--Ghanam--Thompson metrics constructed from an arbitrary number of $n_0$-blocks, where $H$ does not depend on $v$ and spacelike coordinates of the Lorentzian block (then $\Box^{(0)}H=0$).
One can find examples of such spaces e.g.\ by choosing $H=\sum_{\alpha=1}^{N-1} H^{(\alpha)}$, where each $H^{(\alpha)} $
is harmonic on their respective piece, $\Box^{(\alpha)} H^{(\alpha)}=0$, and does not depend on the other coordinates.

\section*{Acknowledgments}
V.P. and A.P. acknowledge support from research plan RVO: 67985840 and research
grant GA\v CR 13-10042S. 

\appendix

\section{Weyl tensor of the product manifolds}
\label{Sec_Weylprod}
Let us consider an $n$-dimensional manifold $M$ constructed as a direct product of $N$ 
$n_\alpha$-dimensional maximally symmetric manifolds $M_\alpha$, $\alpha=0,\dots,N-1$,  $n=n_0+n_1+\dots+n_{N-1}$.
In order to derive frame components of  the Weyl tensor, let us start with $N=2$, i.e. first, let us consider an $n=n_0+n_1$ dimensional Lorentzian manifold $M$, which is a direct product of two maximally symmetric manifolds $M_0$ and $M_1$. $M_0$ is $n_0$-dimensional and Lorentzian with the Ricci scalar $R_0$.
$M_1$ is $n_1$-dimensional and Riemannian with the Ricci scalar $R_1$. As in \cite{PraPraOrt07}, we adapt a frame to the natural product structure with 
\be
g_{ab} = 2 \ell_{(a} n_{b)} +  \delta_{i_0 j_0} m^{(i_0)}_a m^{(j_0)}_b 
+ \delta_{i_1 j_1} m^{(i_1)}_a m^{(j_1)}_b ,
 \ee
where $i_0,j_0=2\dots n_0-1$ and $i_1,j_1= n_0 \dots n-1$.

First, let us observe that $M$ is Einstein iff \eqref{Einstcond} holds. 
Even though  $M_0$ and $M_1$ are conformally flat, $M$ in general is not, since b.w. zero terms of the Weyl tensor are non-trivial.
Using results of \cite{PraPraOrt07} (already assuming  that  \eqref{Einstcond} holds),  we obtain  the following non-vanishing components of the Weyl tensor
\beq
\Phi&=&-\frac{n_1 R_0 }{n_0(n_0-1)(n-1)},\nonumber\\
\Phi_{i_0 j_0}&=&\frac{n_1 R_0 }{n_0(n_0-1)(n-1)}\delta_{i_0j_0},\nonumber\\
\Phi_{i_1j_1 }&=&-\frac{R_0 }{n_0(n-1)}\delta_{i_1j_1},\nonumber\\
\Phi_{i_0j_0k_0l_0 }&=&\frac{n_1 R_0}{n_0(n_0-1)(n-1)} 2\delta_{i_0[k_0}\delta_{l_0]j_0},\nonumber\\
\Phi_{i_1 j_1 k_1 l_1}&=&\frac{R_0}{(n_1-1)(n-1)} 2\delta_{i_1 [k_1}\delta_{l_1 ]j_1},\nonumber\\
\Phi_{i_0 j_1 k_0 l_1}&=&-\frac{R_0 }{n_0(n-1)} \delta_{i_0k_0}\delta_{j_1l_1}.
\eeq

Now, we demand that $M$ is also a vacuum solutions to the quadratic gravity, i.e.\    
$S^{(2)}_{ab} = C_{acde}{C_b}^{cde}=Kg_{ab}$ (see \eqref{QGterm}).
 Components of  $S^{(2)}_{ab}$  are
\beq
S^{(2)}_{01}&=&\frac{2n_1R_0^2}{n_0^2(n_0-1)(n-1)},\nonumber\\ 
S^{(2)}_{i_0 j_0}&=&\frac{2n_1R_0^2}{n_0^2(n_0-1)(n-1)}\delta_{i_0j_0},\nonumber \\
S^{(2)}_{j_1j_1}&=&\frac{2R_0^2}{n_0(n_1-1)(n-1)}\delta_{i_1j_1}. \label{direct2c}
\eeq
For \eqref{direct2c} to be compatible with \eqref{QGterm}, 
\be
n_1=n_0,\ \ \ R_1=R_0.
\ee

Thus, the Weyl tensor frame components read
\beq
\Phi&=&-\frac{R_0}{(n_0-1)(n-1)},\nonumber\\
\Phi_{i_0j_0}&=&\frac{R_0}{(n_0-1)(n-1)}\delta_{i_0j_0},\nonumber\\
\Phi_{i_1j_1}&=&-\frac{R_0 }{n_0(n-1)}\delta_{i_1j_1},\nonumber\\
\Phi_{i_0j_0k_0l_0}&=&\frac{R_0}{(n_0-1)(n-1)} 2\delta_{i_0[k_0}\delta_{l_0]j_0},\nonumber\\
\Phi_{i_1j_1k_1l_1}&=&\frac{R_0}{(n_0-1)(n-1)} 2\delta_{i_1[k_1}\delta_{l_1]j_1},\nonumber\\
\Phi_{i_0j_1k_0l_1}&=&-\frac{R_0 }{n_0(n-1)} \delta_{i_0k_0}\delta_{j_1l_1}.
\eeq

If for a  direct 
product of $p+1$ manifolds $M_0=[n_0,R_0],\  \dots,\  M_{p}=[n_{p},R_{p}]$,
$n_0=\dots=n_{p}$ and $R_0=\dots=R_{p}$ then the product is Einstein, satisfies 
\eqref{QGterm} and the  Weyl frame components have the form
\beq
\Phi^{(p)}&=&-\frac{pR_0}{(n_0-1)(n-1)},\nonumber\\
\Phi_{i_0j_0}^{(p)}&=&\frac{pR_0}{(n_0-1)(n-1)}\delta_{i_0j_0}, \nonumber \\ 
\Phi_{i_\alpha j_\alpha}^{(p)}&=&-\frac{R_0 }{n_0(n-1)}\delta_{i_\alpha j_\alpha},\ \ \alpha \not= 0, \nonumber \\
\Phi_{i_\alpha j_\alpha k_\alpha l_\alpha}^{(p)}&=&\frac{pR_0}{(n_0-1)(n-1)} 2\delta_{i_\alpha[k_\alpha}\delta_{l_\alpha]j_\alpha},\nonumber\\
\Phi_{i_\alpha j_\beta k_\alpha l_\beta}^{(p)}&=&-\frac{R_0 }{n_0(n-1)} \delta_{i_\alpha k_\alpha}\delta_{j_\beta l_\beta},\ \ \alpha\not= \beta .
\label{PhiIJprodkplus}
\eeq
 Let us prove this by the mathematical induction.

We assume that we have $p$ manifolds $M_0=[n_0,R_0],\  \dots,\  M_{p-1}=[n_{p-1},R_{p-1}]$ with $n_0=\dots=n_{p-1}$, $R_0=\dots=R_{p-1}$ 
and $n=pn_0$ with
\beq
\Phi^{(p-1)}&=&-\frac{(p-1)R_0}{(n_0-1)(n-1)},\nonumber\\
\Phi_{i_0j_0}^{(p-1)}&=&\frac{(p-1)R_0}{(n_0-1)(n-1)}\delta_{i_0j_0}, \nonumber\\ 
\Phi_{i_\alpha j_\alpha}^{(p-1)}&=&-\frac{R_0 }{n_0(n-1)}\delta_{i_\alpha j_\alpha},\ \ \alpha \not=0  \nonumber \\
\Phi_{i_\alpha j_\alpha k_\alpha l_\alpha}^{(p-1)}&=&\frac{(p-1)R_0}{(n_0-1)(n-1)} 2\delta_{i_\alpha[k_\alpha}\delta_{l_\alpha]j_\alpha},\nonumber\\
\Phi_{i_\alpha j_\beta k_\alpha l_\beta}^{(p-1)}&=&-\frac{R_0 }{n_0(n-1)} \delta_{i_\alpha k_\alpha}\delta_{j_\beta l_\beta},\ 
\ \alpha\not= \beta ,\label{PhiIJprodk}
\eeq
where $\alpha, \beta = 0,\dots, p-1$ (if not stated otherwise)
and we add $M_{p}=[n_{p},R_{p}]$ with frame indices denoted by $i_{p}$ 
\beq
\Phi^{(p)}&=&-\frac{R_0 [(p-1)n_0+n_{p}]}{n_0(n_0-1)(n-1)},\nonumber\\
\Phi_{i_0j_0}^{(p)}&=&\frac{R_0 [(p-1)n_0+n_{p}]}{n_0(n_0-1)(n-1)}\delta_{i_0j_0},\nonumber\\
\Phi_{i_\alpha j_\alpha}^{(p)}&=&-\frac{R_0 }{n_0(n-1)}\delta_{i_\alpha j_\alpha},\ \ \alpha \not=0 \nonumber\\
\Phi_{i_\alpha j_\alpha k_\alpha l_\alpha}^{(p)}&=&
\frac{[(p-1)n_0+n_{p}]R_0}{n_0(n_0-1)(n-1)} 2\delta_{i_\alpha[k_\alpha}\delta_{l_\alpha]j_\alpha},\ \ \alpha\not=p\nonumber\\
\Phi_{i_\alpha j_\beta k_\alpha l_\beta}^{(p)}&=&-\frac{R_0 }{n_0(n-1)} \delta_{i_\alpha k_\alpha}\delta_{j_\beta l_\beta},\ 
\ \alpha\not= \beta,\ \nonumber\\
\Phi_{i_p j_p k_p l_p}^{(p)}&=&\frac{pR_0}{(n_{p}-1)(n-1)} 2\delta_{i_p[k_p}\delta_{l_p]j_p}, 
\eeq
where now $n=p n_0+n_p$.

Considering the second-rank tensor 
$S^{(2)}_{ab} \equiv C_{acde}{C_b}^{cde}=Kg_{ab}$ (see \eqref{QGterm}),
we get
\beq
S^{(2)}_{01}&=&\frac{2[(p-1)n_0+n_{p}]R_0^2}{n_0^2(n_0-1)(n-1)},\nonumber\\
S^{(2)}_{i_p j_p}&=&\frac{2pR_0^2}{n_0(n_{p}-1)(n-1)}\delta_{i_p j_p},
\eeq
which implies
\be
n_0=\dots=n_{p},\ \ \ R_0=\dots=R_{p} \label{equalnR}
\ee
and \eqref{PhiIJprodkplus} follows.

\section{The generalization of the Khlebnikov--Ghanam--Thompson metric consting of  2-blocks}
\label{app_GT}

In this section, we derive  Christoffel symbols,  Ricci coefficients and Weyl coordinate and frame components for the higher-dimensional generalization
of the Khlebnikov--Ghanam--Thompson metric consisting of $N$ 2-blocks \eqref{GT-2blocks}, considered in section \ref{sec_GT}.

The non-trivial Christoffel symbols for the metric \eqref{GT-2blocks} 
read
\begin{align}
  & \Gamma^u_{uu} = \lambda v, \ \ \Gamma^v_{uu} = \frac{1}{2} H_{,u} + \lambda v (\lambda v^2 + H), \ \
 \Gamma^v_{uv} = \lambda v, \ \
 \Gamma^v_{u x_\alpha} = \frac{1}{2} H_{,x_\alpha}, \ \
 \Gamma^v_{u y_\alpha} = \frac{1}{2} H_{,y_\alpha}, \nonumber\\
 &\Gamma^{x_\alpha}_{uu} = - \frac{|\lambda|}{2} H_{,x_\alpha}, \ \
   \Gamma^{y_\alpha}_{uu} = - \frac{|\lambda|}{2 \s^2(x_\alpha)} H_{,y_\alpha}, \ \
   \Gamma^{x_\alpha}_{y_\beta y_\gamma} = - \s(x_\alpha) \cs(x_\alpha) \delta_{\alpha\beta\gamma}, \ \
   \Gamma^{y_\alpha}_{x_\beta y_\gamma} = \ct(x_\alpha) \delta_{\alpha\beta\gamma},
\end{align}
where $\delta_{\alpha\beta\gamma}=1$ for $\alpha=\beta=\gamma$, i.e.\ all indices correspond to the same block, 
and otherwise $\delta_{\alpha\beta\gamma}$ vanishes.
The independent non-trivial components of the Riemann and Ricci tensors read
\begin{align}
  & R_{uvuv} = - \lambda, \ \
   R_{u x_\alpha u x_\beta} = - \frac{1}{2} H_{,x_\alpha x_\beta}, \ \
   R_{u x_\alpha u y_\beta} = - \frac{1}{2} H_{,x_\alpha y_\beta} + \frac{1}{2} \ct(x_\alpha) H_{,y_\alpha} \delta_{\alpha\beta}, \nonumber \\
  & R_{u y_\alpha u y_\beta} = - \frac{1}{2} H_{,y_\alpha y_\beta} 
	- \frac{1}{2} \s(x_\alpha) \cs(x_\alpha) H_{,x_\alpha} \delta_{\alpha\beta}, \ \
   R_{x_\alpha y_\beta x_\gamma y_\delta} = \frac{\s^2(x_\alpha)}{\lambda} \delta_{\alpha\beta\gamma\delta} ,
\end{align}
\begin{align}
   R_{uu} = \lambda (\lambda v^2 + H) - \frac{1}{2} \Box H, \ \ R_{uv} = \lambda, \ \
   R_{x_\alpha x_\beta} = \sgn(\lambda) \delta_{\alpha\beta}, \ \
   R_{y_\alpha y_\beta} = \sgn(\lambda) \s^2(x_\alpha) \delta_{\alpha\beta},
\end{align}
where
\begin{equation}
  \Box H = |\lambda| \sum_{\alpha=1}^{N-1} \left( H_{,x_\alpha x_\alpha} + \frac{1}{\s^2(x_\alpha)} H_{,y_\alpha y_\alpha} + \ct(x_\alpha) H_{,x_\alpha} \right).
\end{equation}
The Ricci scalar is constant $R = n \lambda$. The metric \eqref{GT-2blocks} 
is Einstein, $R_{ab} = \lambda g_{ab}$, iff $H(u, x_\alpha, y_\alpha)$ is a harmonic function.

The non-trivial components of the Weyl tensor 
are
\begin{align}
  & C_{uvuv} = - \frac{n - 2}{n - 1} \lambda, \nonumber\\
  & C_{u x_\alpha u x_\beta} = - \frac{1}{2} H_{,x_\alpha x_\beta} - \sgn(\lambda) \left( \frac{\lambda v^2 + H}{n - 1}
    - \frac{\Box H}{2 \lambda (n - 2)} \right) \delta_{\alpha\beta}, \nonumber\\
  & C_{u x_\alpha u y_\beta} =  \frac{1}{2}\left( - H_{,x_\alpha y_\beta} 
	+  \ct(x_\alpha) H_{,y_\alpha} \delta_{\alpha\beta}\right), \nonumber\\
  & C_{u y_\alpha u y_\beta} = - \frac{1}{2} H_{,y_\alpha y_\beta} 
	- \s^2(x_\alpha) \left[ \frac{1}{2} \ct(x_\alpha) H_{,x_\alpha}
    + \sgn(\lambda) \left( \frac{\lambda v^2 + H}{n - 1} - \frac{\Box H}{2 \lambda (n - 2)} \right) \right] \delta_{\alpha\beta}, \nonumber\\
  & C_{u x_\alpha v x_\beta} = - \frac{\sgn(\lambda)}{n - 1} \delta_{\alpha\beta}, \nonumber\\
  & C_{u y_\alpha v y_\beta} = - \s^2(x_\alpha) \frac{\sgn(\lambda)}{n - 1} \delta_{\alpha\beta}, \nonumber\\
  & C_{x_\alpha x_\beta x_\gamma x_\delta} = \frac{2}{\lambda (n - 1)} \delta_{\alpha[\delta} \delta_{\gamma]\beta}, \nonumber\\
  & C_{x_\alpha y_\beta x_\gamma y_\delta} = \frac{\s^2(x_\beta)}{\lambda} (\delta_{\alpha\beta\gamma\delta} 
	- \frac{1}{n - 1} \delta_{\alpha\gamma} \delta_{\beta\delta}), \nonumber\\
  & C_{y_\alpha y_\beta y_\gamma y_\delta} = \frac{2 \s^2(x_\alpha) \s^2(x_\beta)}{\lambda (n - 1)} 
	\delta_{\alpha[\delta} \delta_{\gamma]\beta}.
\end{align}

One may introduce a parallelly propagated frame
\begin{align}
  & \ell^a = (0, 1, 0, \ldots, 0), \nonumber\\
  & n^a = (1, - \frac{1}{2} g_{uu}, 0, \ldots, 0), \nonumber\\
  & m^a_{(\chi_\alpha)} = (0, \dots, 0, \overbrace{\sqrt{|\lambda|}, 0}^{\text{$\alpha^{\mbox{th}}$ block}}, 0, \ldots, 0), \nonumber\\  
  & m^a_{(\upsilon_\alpha)} = (0, \dots, 0, \underbrace{0, \frac{\sqrt{|\lambda|}}{\s(x_\alpha)}}_{\text{$\alpha^{\mbox{th}}$ block}}, 0, \ldots, 0)
	\label{frame-2}
\end{align}
to show that only b.w.\ 0 and $-2$ components of the Weyl tensor are non-vanishing. Moreover, the b.w.\ 0
components are constant. 

One can use the relation between the Riemann and Weyl tensor \eqref{RiemannWeyl} and 
 \eqref{RiemannGT} to obtain the b.w.\ 0 components of the Weyl tensor for arbitrary N $n_0$-blocks 
\BEA
{C^{ab}}_{cd}&=&\Lambda \sum_{\alpha=0}^{N-1}\left[ (\delta^{(\alpha)})^a_{~c}(\delta^{(\alpha)})^b_{~d}
-(\delta^{(\alpha)})^a_{~d}(\delta^{(\alpha)})^b_{~c}\right]
\nonumber\\
&&-\frac{\Lambda (n_0-1)}{n-1}\left[ \sum_{\alpha=0}^{N-1} (\delta^{(\alpha)})^a_{~c}\sum_{\beta=0}^{N-1} (\delta^{(\beta)})^b_{~d}
-\sum_{\alpha=0}^{N-1}(\delta^{(\alpha)})^a_{~d}\sum_{\beta=0}^{N-1}(\delta^{(\beta)})^b_{~c}\right].\label{C0}
\EEA
By contracting \eqref{C0} with the frame vectors, one obtains the frame components
\BEA
\Phi&=&{C^{ab}}_{cd}\ell_a n_b \ell^c n^d=-\Lambda\left(1-\frac{n_0-1}{n-1}\right)=-\frac{n-n_0}{n-1}\Lambda=
-\frac{n_0(N-1)}{n-1}\Lambda,\label{Phi_max}\nonumber\\
\Phi_{i_0 j_0 }&=&{C^{ab}}_{cd}\ell_a 
m^{(i_0)}_b n^c m^{(j_0)d} =\frac{n_0(N-1)}{n-1}\Lambda\delta_{i_0 j_0}=-\Phi\delta_{i_0 j_0 },\nonumber\\
\Phi_{i_\alpha j_\beta }&=&{C^{ab}}_{cd}\ell_a 
m^{(i_\alpha)}_b n^c m^{(j_\beta)d}=-\Lambda\frac{n_0-1}{n-1}\delta_{i_\alpha j_\beta},\ \ \alpha,\beta \not=0,\nonumber\\
\Phi_{i_\alpha j_\beta k_\gamma l_\delta }&=&{C^{ab}}_{cd} 
m^{(i_\alpha)}_a m^{(j_\beta)}_b m^{(k_\gamma) c} m^{(l_\delta )d}=
2\Lambda\delta_{i_\alpha  [k_\gamma}\delta_{l_\delta]j_\beta}\left(\delta_{\alpha\beta\gamma\delta}
- 
\frac{n_0-1}{n-1} \right). 
\EEA
Thus $\Phi_{ij}$ is diagonal
\be
\Phi_{ij}=\frac{\Lambda}{n-1}\mbox{diag}\left[n_0(N-1)\delta_{i_0  j_0 },-(n_0-1)\delta_{i_\alpha j_\beta } \right], \alpha,\beta\not=0.\label{Phijkl_max}
\ee
For $n_0=2$, $\Lambda=\lambda/(n_0-1)=\lambda$, one gets non-vanishing b.w. 0 components
\begin{align}
  & \Phi_{\chi_\alpha \chi_\beta \chi_\gamma \chi_\delta} = 
	\Phi_{\upsilon_\alpha \upsilon_\beta \upsilon_\gamma \upsilon_\delta} =\frac{2 \lambda}{(n - 1)} \delta_{\alpha[\delta} \delta_{\gamma]\beta}, \nonumber \\
  & \Phi_{\chi_\alpha \upsilon_\beta \chi_\gamma \upsilon_\delta} = \lambda (\delta_{\alpha\beta\gamma\delta} 
	- \frac{1}{n - 1} \delta_{\alpha\gamma} \delta_{\beta\delta}), \nonumber\\
  & \Phi_{\chi_\alpha \chi_\beta} = \Phi_{\upsilon_\alpha \upsilon_\beta} =-\frac{\lambda}{n - 1} \delta_{\alpha\beta},\ \ 
   \Phi = -\frac{n - 2}{n - 1} \lambda. \label{Weyl_2b}
\end{align}
The b.w.  $-1$ components, $\Psi'_{ijk}$, vanish and b.w. $-2$ components 
\begin{align}
  & \Omega'_{\chi_\alpha \chi_\beta} = -\frac{|\lambda|}{2} H_{,x_\alpha x_\beta} + \frac{\Box H}{2 (n - 2)} \delta_{\alpha\beta}, \nonumber\\
  & \Omega'_{\chi_\alpha \upsilon_\beta} = -\frac{|\lambda|}{2 \s(x_\beta)} H_{,x_\alpha y_\beta} 
	+ \frac{|\lambda| \ct(x_\alpha)}{2 \s(x_\alpha)} H_{,y_\alpha} \delta_{\alpha\beta}, \nonumber\\
  & \Omega'_{\upsilon_\alpha \upsilon_\beta} = -\frac{|\lambda|}{2 \s(x_\alpha) \s(x_\beta)} H_{,y_\alpha y_\beta}
	- \frac{1}{2} \left( |\lambda| \ct(x_\alpha) H_{,x_\alpha}
    + \frac{\Box H}{(n - 2)} \right) \delta_{\alpha\beta}
\end{align}
are independent of $v$ and thus $D\Omega'_{ij}=0$.

The Ricci rotation coefficients read
\begin{align}
\kappa_{i}& =L_{10}= \tau'_i 
= {\stackrel{i}{M}}_{j0} = 0, \ \
 \rho_{ij}  = L_{1i} = \tau_i =\rho'_{ij}=\Mi_{j1}= 0, \ \
 L_{11}  = \lambda v, \nonumber\\
\kappa'_{\chi_\alpha} & =-\frac{\sqrt{|\lambda|}}{2} H,_{x_\alpha},\ \
\kappa'_{v_\alpha }  =-\frac{\sqrt{|\lambda|}}{2s(x_\alpha)} H,_{y_\alpha},\ \
\Mchal_{v_\beta v_\gamma} =ct(x_\alpha)  
\sqrt{|\lambda|}\ \delta_{\alpha\beta\gamma},\label{Ricii_coeff_2b}
\end{align}
the remaining spin coefficients $\Mi_{jk}$ vanish.
From \eqref{Ricii_coeff_2b}, it follows  that the frame \eqref{frame-2} is parallelly propagated and that the metric
\eqref{GT-2blocks} 
belongs to the RNV subclass of the Kundt spacetimes.

Now, let us study components of $\nabla^{(1)}C$. 
In the b.w.\ 0 components of $\nabla^{(1)}C$, terms like
\be 
\Phi_{ij}\Mi_{kl}+\Phi_{ki}\Mi_{jl}\label{block2-derCA}
\ee
 and 
\be
\Phi_{sjkl}\Ms_{ih}+\Phi_{iskl}\Ms_{jh}+\Phi_{ijsl}\Ms_{kh}+\Phi_{ijks}\Ms_{lh}\label{block2-derCB}
\ee
appear and they all vanish due to the form of $\Mi_{jk}$ \eqref{Ricii_coeff_2b} and $\Phi_{ij}$ and $\Phi_{ijkl}$ \eqref{Weyl_2b}.
Since $\tau_i=\rho'_{ij}=\Mi_{j1}=D\Omega'_{ij}=0$ and $\Phi_{ijkl}$ are constant, the only terms contributing to 
b.w.\ $-1$ components of $\nabla^{(1)}C$ are 
\be
\ell_{a;b}=L_{11}\ell_a\ell_b,\ \ n_{a;b}=-L_{11}n_a\ell_{b}.\label{block2-derCC}
\ee
However, in the expression of the b.w.\ 0 part of the Weyl tensor \eqref{eq:rscalars},
 there is same number of frame vectors $\bl$'s and $\bn$'s and thus such terms cancel out, 
e.g. $C_{abcd;e}=\Phi L_{11}(1-1+1-1)\ell_a n_b\ell_c n_d \ell_e+\dots$.

Thus the metric \eqref{GT-2blocks}  obeys the assumptions of  proposition \ref{lemma_derbalanced}.


\section{The generalization of the Khlebnikov--Ghanam--Thompson metric consisting of 3-blocks}
\label{app_GT_3}

In this section, we derive Christoffel symbols, Ricci coefficients and Weyl coordinate and frame components 
for another higher-dimensional generalization
of the Khlebnikov--Ghanam--Thompson metric consisting of $N$ 3-blocks \eqref{appGT3:GTmetric}, considered in section \ref{sec_GT_3}.

The non-trivial Christoffel symbols for the metric \eqref{appGT3:GTmetric} are given by
\begin{align}
  & \Gamma^u_{uz} = -\frac{1}{z}, \qquad
  \Gamma^v_{uu} = \frac{1}{2} \left(H_{,u} - \lambda vz H_{,z} \right), \qquad
  \Gamma^v_{uv} = \lambda v, \nonumber\\
  & \Gamma^v_{uz} = \frac{1}{2} H_{,z} + \frac{1}{z}(2 \lambda v^2 + H), \qquad
  \Gamma^v_{u x_\alpha} = \frac{1}{2} H_{,x_\alpha}, \qquad
  \Gamma^v_{u y_\alpha} = \frac{1}{2} H_{,y_\alpha}, \qquad
  \Gamma^v_{u z_\alpha} = \frac{1}{2} H_{,z_\alpha}, \nonumber\\
  & \Gamma^v_{vz} = \frac{1}{z}, \qquad
  \Gamma^z_{uu} = \frac{\lambda z^2}{4} H_{,z}, \qquad
  \Gamma^z_{uv} = - \frac{\lambda z}{2}, \qquad
  \Gamma^z_{uz} = - \lambda v, \qquad
  \Gamma^z_{zz} = - \frac{1}{z}, \nonumber\\
  & \Gamma^{x_\alpha}_{uu} = \frac{\lambda}{4} H_{,x_\alpha}, \qquad
  \Gamma^{x_\alpha}_{y_\beta y_\gamma} = - sh_\alpha ch_\alpha 
	\, \delta_{\alpha\beta\gamma}, \qquad
  \Gamma^{x_\alpha}_{z_\beta z_\gamma} = - sh_\alpha ch_\alpha s^2_\alpha  
	\, \delta_{\alpha\beta\gamma}, \nonumber\\
  & \Gamma^{y_\alpha}_{uu} = \frac{\lambda}{4 sh_\alpha^2 
	} H_{,y_\alpha}, \qquad
  \Gamma^{y_\alpha}_{x_\beta y_\gamma} = 
	cth_\alpha\, 
	\delta_{\alpha\beta\gamma}, \qquad
  \Gamma^{y_\alpha}_{z_\beta z_\gamma} = - s_\alpha c_\alpha 
	\, \delta_{\alpha\beta\gamma}, \nonumber\\
  & \Gamma^{z_\alpha}_{uu} = \frac{\lambda}{4 sh^2_\alpha s^2_\alpha 
	} H_{,z_\alpha}, \qquad
  \Gamma^{z_\alpha}_{x_\beta z_\gamma} = 
	cth_\alpha\, \delta_{\alpha\beta\gamma}, \qquad
  \Gamma^{z_\alpha}_{y_\beta z_\gamma} = 
	ct_\alpha\, \delta_{\alpha\beta\gamma},
\end{align}
where $ch_\alpha=\cosh x_\alpha$,  $c_\alpha=\cos y_\alpha$, $cth_\alpha=ch_\alpha/sh_\alpha$ and $ct_\alpha =c_\alpha / s_\alpha$.
The  non-zero components of the Riemann tensor read
\begin{align}
  & R_{uvuv} = - \frac{\lambda}{2}, \ \
   R_{uvuz} = - \frac{\lambda v}{z}, \ \
   R_{u z u z} = - \frac{1}{2 z^2} (4 \lambda v^2 + 2H + 3z H_{,z} + z^2 H_{,zz}), \nonumber\\
  & R_{u z u x_\alpha} = - \frac{1}{2z} (H_{,x_\alpha} + z H_{,z x_\alpha}), \ \
   R_{u z u y_\alpha} = - \frac{1}{2z} (H_{,y_\alpha} + z H_{,z y_\alpha}), \ \
   R_{u z u z_\alpha} = - \frac{1}{2z} (H_{,z_\alpha} + z H_{,z z_\alpha}), \nonumber\\
  & R_{u z v z} = - \frac{1}{z^2}, \ \
   R_{u x_\alpha u x_\beta} = - \frac{1}{2} H_{,x_\alpha x_\beta}, \ \
   R_{u x_\alpha u y_\beta} = \frac{1}{2} \left( -H_{,x_\alpha y_\beta} 
	+ cth_\alpha 
	H_{,y_\alpha} \delta_{\alpha\beta}\right), \nonumber\\
  & R_{u x_\alpha u z_\beta} =  \frac{1}{2}\left( - H_{,x_\alpha z_\beta} + cth_\alpha 
	H_{,z_\alpha} \delta_{\alpha\beta}\right), \ \
   R_{u y_\alpha u y_\beta} = - \frac{1}{2} \left( H_{,y_\alpha y_\beta} + sh_\alpha ch_\alpha 
	H_{,x_\alpha} \delta_{\alpha\beta}\right), \nonumber\\
  & R_{u y_\alpha u z_\beta} = \frac{1}{2}\left( - H_{,y_\alpha z_\beta} +ct_\alpha 
	H_{,z_\alpha} \delta_{\alpha\beta}\right), \ \
   R_{u z_\alpha u z_\beta} = - \frac{1}{2}\left[ H_{,z_\alpha z_\beta} +  s_\alpha 
	(c_\alpha 
	H_{,y_\alpha} + sh_\alpha ch_\alpha s_\alpha 
	H_{,x_\alpha}) \delta_{\alpha\beta}\right], \nonumber\\
  & R_{x_\alpha y_\beta x_\gamma y_\delta} = \frac{2 sh^2_\alpha 
	}{\lambda} \delta_{\alpha\beta\gamma\delta}, \qquad
  R_{x_\alpha z_\beta x_\gamma z_\delta} = \frac{2 sh^2_\alpha s^2_\alpha 
	}{\lambda} \delta_{\alpha\beta\gamma\delta}, \ \
   R_{y_\alpha z_\beta y_\gamma z_\delta} = \frac{2 sh^4_\alpha s^2_\alpha 
	}{\lambda} \delta_{\alpha\beta\gamma\delta}
\end{align}
and those of the Ricci tensor
\begin{align}
  & R_{uu} = \lambda \left( H + z H_{,z} \right) - \frac{1}{2} \Box H, \ \
   R_{uv} = \lambda, \ \
   R_{uz} = \frac{2 \lambda v}{z}, \ \
   R_{zz} = - \frac{2}{z^2}, \nonumber\\
  & R_{x_\alpha x_\beta} = -2 \delta_{\alpha\beta}, \ \
   R_{y_\alpha y_\beta} = -2 sh^2_\alpha 
	\delta_{\alpha\beta}, \ \
   R_{z_\alpha z_\beta} = -2sh^2_\alpha s^2_\alpha 
	\delta_{\alpha\beta},
\end{align}
where
\begin{equation}
  \Box H\! = \! -\frac{\lambda}{2} \left[ z^2 H_{,zz} - z H_{,z} + \!\!\sum_{\alpha=1}^{N-1}\!\! \left( H_{,x_\alpha x_\alpha}
	+ 2 cth_\alpha 
	H_{,x_\alpha}
  + \frac{1}{sh^2_\alpha 
	} \left( H_{,y_\alpha y_\alpha} + ct_\alpha 
	H_{,y_\alpha} + \frac{1}{s^2_\alpha 
	} H_{,z_\alpha z_\alpha} \right) \right) \right].
\end{equation}
The Ricci scalar is constant $R = n \lambda$. The metric \eqref{appGT3:GTmetric} is Einstein with $R_{ab} = \lambda g_{ab}$ iff
\begin{equation}
  \Box H = 2 \lambda z H_{,z}.
\end{equation}

Let us choose a parallelly propagated frame
\begin{align}
  & \ell^a = (0, 1, 0, \ldots, 0),\nonumber \\
  & n^a = (1, - \frac{3 \lambda v^2 + 2 H}{4}, \frac{\lambda v z}{2}, 0, \ldots, 0) ,\nonumber\\
  & m_{(2)}^a = (0, - \sqrt{\frac{-\lambda}{2}} v, \sqrt{\frac{-\lambda}{2}} z, 0, \ldots, 0) ,\nonumber\\
  & m^a_{(\chi_\alpha)} = (0, \dots, 0, \overbrace{\sqrt{\frac{-\lambda}{2}}, 0, 0}^{\text{$\alpha^{\mbox{th}}$ block}}, 0, \ldots, 0),\nonumber \\  
  & m^a_{(\upsilon_\alpha)} = (0, \dots, 0, \underbrace{0, \frac{\sqrt{-\lambda}}{\sqrt{2} sh_\alpha 
	}, 0}_{\text{$\alpha^{\mbox{th}}$ block}}, 0, \ldots, 0) ,\nonumber\\ 
  & m^a_{(\zeta_\alpha)} = (0, \dots, 0, \underbrace{0, 0, \frac{\sqrt{-\lambda}}{\sqrt{2} sh_\alpha s_\alpha 
	}}_{\text{$\alpha^{\mbox{th}}$ block}}, 0, \ldots, 0) 
\end{align}
and express the frame components of the Weyl tensor. As in the 2-block case,
b.w.\ 0 components are constant and can be derived using \eqref{Phi_max}--\eqref{Phijkl_max}
\begin{align}
& \Phi=-\frac{(n-3)\lambda}{2(n-1)},\ \ 
\Phi_{ij}=\frac{\lambda}{n-1}\mbox{diag}\left(\frac{n-3}{2},-\delta_{ij}\right),\nonumber\\
  & \Phi_{2 \chi_\alpha 2 \chi_\beta} =  \Phi_{2 \upsilon_\alpha 2 \upsilon_\beta}
	= \Phi_{2 \zeta_\alpha 2 \zeta_\beta} = - \frac{\lambda}{n - 1} \delta_{\alpha\beta}, \nonumber\\
  & \Phi_{\chi_\alpha \chi_\beta \chi_\gamma \chi_\delta}
	=  \Phi_{\upsilon_\alpha \upsilon_\beta \upsilon_\gamma \upsilon_\delta} 
	= \Phi_{\zeta_\alpha \zeta_\beta \zeta_\gamma \zeta_\delta}
	= - \frac{2 \lambda}{n - 1} \delta_{\alpha[\gamma} \delta_{\delta]\beta}, \nonumber\\
  & \Phi_{\chi_\alpha \upsilon_\beta \chi_\gamma \upsilon_\delta} 
	=  \Phi_{\chi_\alpha \zeta_\beta \chi_\gamma \zeta_\delta} 
	= \Phi_{\upsilon_\alpha \zeta_\beta \upsilon_\gamma \zeta_\delta} 
	= \frac{\lambda}{2} \delta_{\alpha\beta\gamma\delta}
	- \frac{\lambda}{n - 1} \delta_{\alpha\gamma} \delta_{\beta\delta},\label{Weyl_3b}
\end{align}
b.w.\ $-1$ components $\Psi'_{ijk}$ vanish and b.w.\ $-2$ non-vanishing components are independent of $v$
\begin{align}
  & \Omega'_{22} = \frac{1}{2(n - 2)} \left( \Box H - 2 \lambda z H_{,z} \right) + \frac{\lambda z}{4} \left( z H_{,zz} + 3 H_{,z} \right), \nonumber\\
  & \Omega'_{2 \chi_\alpha} = \frac{\lambda}{4} \left( H_{,x_\alpha} + z H_{,z x_\alpha} \right), \ \ 
   \Omega'_{2 \upsilon_\alpha} = \frac{\lambda}{4 sh_\alpha 
	} \left( H_{,y_\alpha} + z H_{,z y_\alpha} \right), \ \ 
   \Omega'_{2 \zeta_\alpha} = \frac{\lambda}{4 sh_\alpha s_\alpha 
	} \left( H_{,z_\alpha} + z H_{,z z_\alpha} \right), \nonumber\\
  & \Omega'_{\chi_\alpha \chi_\beta} 
	= \frac{\lambda}{4} H_{,x_\alpha x_\beta} + \frac{1}{2(n - 2)} \left( \Box H - 2 \lambda z H_{,z} \right) \delta_{\alpha\beta}, \nonumber\\
  & \Omega'_{\upsilon_\alpha \upsilon_\beta} = \frac{\lambda}{4 sh_\alpha  
	} \left( \frac{1}{sh_\beta} H_{,y_\alpha y_\beta} + ch_\alpha 
	H_{,x_\alpha} \delta_{\alpha\beta}\right)
    + \frac{1}{2(n - 2)} \left( \Box H - 2 \lambda z H_{,z} \right) \delta_{\alpha\beta}, \nonumber\\
  & \Omega'_{\zeta_\alpha \zeta_\beta} = \frac{\lambda}{4 sh_\alpha sh_\beta s_\alpha s_\beta 
	} H_{,z_\alpha z_\beta}
    + \frac{\lambda}{4} \left( cth_\alpha 
		H_{,x_\alpha} + \frac{ct_\alpha 
		}{  		sh^2_\alpha 
		} H_{,y_\alpha} \right) \delta_{\alpha\beta}
    + \frac{1}{2(n - 2)} \left( \Box H - 2 \lambda z H_{,z} \right) \delta_{\alpha\beta}, \nonumber\\
  & \Omega'_{\chi_\alpha \upsilon_\beta} = \frac{\lambda}{4 sh_\beta 
	} \left( H_{,x_\alpha y_\beta} 
	- 	cth_\alpha H_{,y_\alpha} \delta_{\alpha\beta}\right), \nonumber\\
  & \Omega'_{\chi_\alpha \zeta_\beta} = \frac{\lambda}{4 sh_\beta s_\beta  
	} \left(H_{,x_\alpha z_\beta} 
	- 	cth_\alpha H_{,z_\alpha} \delta_{\alpha\beta} \right),\nonumber\\
  & \Omega'_{\upsilon_\alpha \zeta_\beta} = \frac{\lambda}{4 sh_\alpha sh_\beta s_\beta 
	} \left(H_{,y_\alpha z_\beta} 
	- 	ct_\alpha H_{,z_\alpha} \delta_{\alpha\beta}\right).
\end{align}

The Ricci rotation coefficients are
	\begin{align}
\kappa_{i} &=L_{10}= \tau'_i 
= {\stackrel{i}{M}}_{j0} = 0, \ \  
\rho_{ij}  =\Mi_{j1}= 0, \ \nonumber\\
L_{11} &= \lambda v, \ \ \tau_i = L_{1i} = \left(\sqrt{\frac{-\lambda}{2}}, 0, \dots, 0 \right), \ \ 
 \rho'_{ij} = \mbox{diag}\left(- \frac{\lambda v}{2}, 0, \dots, 0 \right), \ \ 
\nonumber\\
   {\stackrel{2}{M}}_{jk} &= 0, \ \
   {\stackrel{\chi_\alpha}{M}}_{\upsilon_\beta \upsilon_\gamma} = \sqrt{-\frac{\lambda}{2}}
	cth_\alpha 
	\delta_{\alpha\beta\gamma} = {\stackrel{\chi_\alpha}{M}}_{\zeta_\beta \zeta_\gamma}, \ \
   {\stackrel{\upsilon_\alpha}{M}}_{\zeta_\beta \zeta_\gamma} = \sqrt{-\frac{\lambda}{2}}
	\frac{ct_\alpha 	}{ 	sh_\alpha }\delta_{\alpha\beta\gamma},\nonumber\\
\kappa'_2 &= \frac{1}{2} \sqrt{\frac{-\lambda}{2}} \left( \frac{\lambda}{2} v^2 - H - z H_{,z} \right),\nonumber\\
\kappa'_{\chi_\alpha}& = - \frac{1}{2} \sqrt{\frac{-\lambda}{2}} H_{,x_\alpha}, \ \
\kappa'_{\upsilon_\alpha} = - \frac{1}{2 sh_\alpha} \sqrt{\frac{-\lambda}{2}} H_{,y_\alpha},\ \
\kappa'_{\zeta_\alpha} = - \frac{1}{2 sh_\alpha s_\alpha} \sqrt{\frac{-\lambda}{2}} H_{,z_\alpha}.
	\label{Ricii_coeff_3b}
\end{align}
Note that the metric \eqref{appGT3:GTmetric} is not RNV since $\tau_i$ is non-vanishing.
One can see that $D\Omega'_{ij}$ vanishes also directly from the Bianchi equation (A.17) in \cite{Durkeeetal10}
\be
D\Omega'_{ij}=\Phi\rho'_{ij}+\Phi_{si}\rho'_{sj}=0.\label{DOmega}
\ee

Now, let us study components of $\nabla^{(1)}C$. Similarly as in the 2-block case,
in the b.w.\ 0 components of $\nabla^{(1)}C$, terms like \eqref{block2-derCA}, \eqref{block2-derCB}
appear and they all again vanish due to the form of $\Mi_{jk}$ \eqref{Ricii_coeff_3b} and $\Phi_{ij}$ and $\Phi_{ijkl}$ \eqref{Weyl_3b}.
Moreover, there are b.w.\ 0 terms 
\[
\Phi\tau_i+\Phi_{ij}\tau_j
\]
and
\[
\Phi_{ik}\tau_j-\Phi_{ij}\tau_k-\Phi_{sijk}\tau_s
\]
that vanish as well.
In b.w.\ $-1$  components of $\nabla^{(1)}C$, there are also terms \eqref{block2-derCC} that again cancel 
each other and
terms
\[
\Phi\rho'_{ij}+\Phi_{ki}\rho'_{kj}
\]
and 
\[
\Phi_{ki}\rho'_{jl}-\Phi_{ji}\rho'_{kl}-\Phi_{sijk}\rho'_{sl}
\]
that vanish as well.

Thus the metric \eqref{appGT3:GTmetric} obeys the assumptions of  proposition \ref{lemma_derbalanced}.

\end{document}